\documentclass[11pt,a4paper]{article}
\usepackage[T1]{fontenc}
\usepackage[utf8]{inputenc}
\usepackage{lmodern,microtype,amsmath,amssymb,amsthm,mathtools}
\usepackage[margin=27mm]{geometry}
\usepackage{xcolor,enumitem,booktabs,graphicx}
\usepackage{placeins}
\graphicspath{{figures/}}
\usepackage[colorlinks=true,linkcolor=blue!45!black,citecolor=blue!45!black,urlcolor=blue!45!black]{hyperref}
\hypersetup{pdftitle={The planar Neumann--Poincare inverse spectral problem},pdfauthor={Habib Ammari and Matthew J. Colbrook}}
\newtheorem{theorem}{Theorem}[section]
\newtheorem{lemma}[theorem]{Lemma}
\newtheorem{proposition}[theorem]{Proposition}
\newtheorem{corollary}[theorem]{Corollary}
\theoremstyle{definition}
\newtheorem{definition}[theorem]{Definition}
\newtheorem{remark}[theorem]{Remark}
\newcommand{\R}{\mathbb R}
\newcommand{\C}{\mathbb C}
\newcommand{\N}{\mathbb N}
\newcommand{\Tr}{\operatorname{Tr}}
\newcommand{\diag}{\operatorname{diag}}
\newcommand{\HS}{\mathrm{HS}}
\newcommand{\op}{\mathrm{op}}
\newcommand{\dd}{\,\mathrm d}
\newcommand{\norm}[1]{\left\lVert#1\right\rVert}

\newcommand{\avec}{\mathbf a}
\newcommand{\bvec}{\mathbf b}
\newcommand{\doi}[1]{\href{https://doi.org/#1}{doi:\nolinkurl{#1}}}
\setlist{itemsep=3pt,topsep=5pt}
\title{On the planar Neumann--Poincar\'e inverse spectral problem}
\author{Habib Ammari\footnote{
  ETH Z\"urich, Department of Mathematics, Rämistrasse 101, 8092 Z\"urich, Switzerland. Email: habib.ammari@math.ethz.ch.} \and  Matthew J. Colbrook\footnote{Department of Applied Mathematics and Theoretical Physics, University of Cambridge, UK. Email: m.colbrook@damtp.cam.ac.uk.}}
\date{}

\newcommand{\revisionfloat}{\color{black}\hypersetup{linkcolor=blue!45!black,citecolor=blue!45!black,urlcolor=blue!45!black}}
\begin{document}
\maketitle

\begin{abstract}
Does the Neumann--Poincar\'e spectrum determine a planar domain up to an extended M\"obius transformation? We construct a pair showing that the answer is no, even for analytic, strictly convex, algebraic boundaries. Their operators are similar. We then prove a positive result for prescribed finite-dimensional families of normalised exterior conformal maps. For every $m\ge1$, the first $m$ positive eigenvalues determine $m$ real Laurent coefficients on a small truncated convex cone, provided the higher coefficients are a prescribed analytic remainder of quadratic order. Recovery is Lipschitz stable, uniformly down to the disk. The proof combines the Grunsky representation with a nondegeneracy argument for weighted Hankel matrices and differentiated perturbation estimates. Explicit two- and three-parameter examples, single-mode families with repeated leading eigenvalues, and M\"obius equivalence classes illustrate the result. Direct proofs for complete polarization data show what spectral weights add to the eigenvalues.
\end{abstract}



\section{Introduction}\label{sec:motivation}
The spectrum of the planar Neumann--Poincar\'e (NP) operator is unchanged by extended M\"obius transformations. Do these transformations account for all domains with the same spectrum? We show that they do not, even for analytic, strictly convex domains arbitrarily close to the disk. We also show that finitely many eigenvalues can determine a domain stably within a prescribed finite-dimensional family.

Let $D\subset\R^2$ be a bounded smooth Jordan domain, with outward unit normal $\nu$. Our normalisation is
\begin{equation}\label{eq:NP}
 K_D^*[\varphi](x)=\frac1{2\pi}\int_{\partial D}
 \frac{(x-y)\cdot\nu_x}{|x-y|^2}\varphi(y)\dd s_y,
 \qquad x\in\partial D.
\end{equation}
The single-layer potential has kernel $(2\pi)^{-1}\log|x-y|$. Its exterior and interior normal derivatives are $(\frac12I+K_D^*)\varphi$ and $(-\frac12I+K_D^*)\varphi$. Throughout, by the NP spectrum we mean the spectrum on the mean-zero space $H^{-1/2}_0(\partial D)$, counted with multiplicities. Thus we exclude the eigenvalue $1/2$ of the full operator. In the single-layer energy inner product the reduced operator is compact and self-adjoint, with nonzero eigenvalues in pairs $\mu,-\mu$; see \cite{KPS}. We denote the positive eigenvalues in nonincreasing order by $\lambda_j(D)$, appending zeros if there are only finitely many.

Schiffer connected this spectrum with conformal mapping through his work on Fredholm eigenvalues \cite{Schiffer1957,Schiffer}. The relation between the classical Fredholm parameter and an NP eigenvalue depends on the normalisation of the integral equation. NP eigenvalues also describe quasistatic resonances of dielectric inclusions; see \cite{AK,Photonics}. Recovery and design from partial spectral data were studied in \cite{ACLZ}, and reconstruction from plasmonic resonances for particular families in \cite{APRYZ}. Shape derivatives and real-analytic dependence of symmetric functions of isolated nonzero NP eigenvalue clusters are established in \cite{DRLLM}.

The counterexample is a pair of degree-seven polynomial lemniscates. A factorisation of the difference of their defining polynomials gives a correspondence between the boundaries. Its invertible incidence matrix, combined with explicit smooth multipliers, intertwines the NP operators. An algebraic argument then rules out holomorphic and antiholomorphic M\"obius equivalence. The boundaries are analytic, strictly convex and real algebraic, and approach the circle after rescaling.

For the positive results we fix capacity and translation by taking exterior conformal maps with leading coefficient one and constant coefficient zero. The first $m$ real Laurent coefficients are free parameters. The higher coefficients are prescribed analytic functions vanishing to second order at the disk. For every finite $m$, a convex cone of parameters exists on which the first $m$ positive NP eigenvalues determine the coefficients with uniform Lipschitz estimates up to the vertex. The coordinate system is fixed, and uniqueness is asserted within this prescribed family.

The proof begins with the Grunsky operator. Its linear part is a finite weighted Hankel matrix, and we find directions in which the characteristic coefficients have an invertible differential and the singular values are distinct and nonzero. The nonlinear remainder is quadratic in Hilbert--Schmidt norm, with a derivative of first order. These estimates control the derivatives of the leading eigenvalues even though the spectral gaps shrink linearly at the disk. On a narrow convex cone the spectral differential stays close to one fixed invertible matrix. Integration along line segments then gives injectivity and the inverse Lipschitz estimate throughout the cone.

The resulting estimates compare parameter distance, weighted Laurent coefficient norms, Hausdorff distance, and distances between complete ordered spectra. Forward estimates hold on a small parameter ball; inverse estimates require the nondegenerate cone. For two coefficients, an explicit inverse of the leading spectral map gives a full sector. For three, exact root bounds give a cone about $(1,1,1)$. We also treat ellipses, single-mode families and M\"obius-invariant classes.

Why, then, can polarization measurements determine a domain? They contain more than eigenvalues. For a harmonic polynomial $H$ and a polynomial $G$, consider
\begin{equation}\label{eq:GPTmotivation}
 M_{H,G}(\lambda;D)=\int_{\partial D}G
 (\lambda I-K_D^*)^{-1}[\partial_\nu H]\dd s,
\end{equation}
where the inverse acts on $H^{-1/2}_0(\partial D)$ and $\lambda$ lies outside the reduced spectrum. A scalar conductivity contrast $k\ne1$ gives $\lambda=(k+1)/(2(k-1))$. The expansion of \eqref{eq:GPTmotivation} contains eigenvalues together with geometric excitation--observation weights. For positive conductivities, Ammari and Kang \cite[Theorem~4.9]{AK} recover the inclusion and contrast from all generalised polarization tensors (GPTs). Their proof recovers a Dirichlet-to-Neumann map and uses inverse conductivity uniqueness \cite{isakov}. Uniqueness from harmonic combinations alone is stated in \cite[Theorem~3.1]{APSTident}; identification and reconstruction of algebraic boundaries from finitely many GPTs are studied in \cite{APSTident,APSTrecon}. We give two direct proofs for complete data: weighted spectral data determine all mixed area moments, and full uncontracted GPTs at one fixed contrast determine an analytic Jordan domain and its contrast. The second proof uses all polynomial observations, a richer data set than harmonic observations alone. Spectral uniqueness and polarization uniqueness therefore ask different questions.

Section~\ref{sec:counter} constructs the counterexample. Section~\ref{sec:class} proves finite-dimensional recovery and stability, and Section~\ref{sec:examples} works out the two- and three-parameter cases. Section~\ref{sec:conclusion} discusses questions left open. Appendix~\ref{app:moments} studies the two-moment map and its critical curve; Appendix~\ref{app:direct} gives the polarization uniqueness proofs; Appendix~\ref{sec:enlarge} develops further uniqueness classes.

\section{An isospectral pair of polynomial lemniscates}\label{sec:counter}

A degree-seven polynomial correspondence supplies the counterexample. The proof has three parts: establish the geometry of the lemniscates, construct a similarity of their NP operators, and rule out extended M\"obius equivalence. Polynomial difference factorisations and correspondences are treated in \cite{CNC,Fried}; transplantation for planar Dirichlet spectra is developed in \cite{BCDS}.

An extended Möbius transformation is a Möbius map, possibly followed by complex conjugation. We require its pole to lie off the boundary and take the transformed domain to be the bounded component enclosed by the image boundary.

Let
\[
 d=i\sqrt7,\qquad \alpha=-\frac{1+d}{2},\qquad
 \beta=\frac{-5+d}{14},
\]
and form
\begin{equation}\label{eq:PQ}
 P(z)=z^7+\alpha z^4+\beta z+i,
 \qquad Q(z)=z^7+\bar\alpha z^4+\bar\beta z+i.
\end{equation}
For $R\ge4096$, take $D_P=\{|P|<R\}$ and $D_Q=\{|Q|<R\}$.

\begin{theorem}[Isospectral polynomial lemniscates]\label{thm:counter}
The domains $D_P,D_Q$ have analytic, strictly convex Jordan boundaries. Their full NP operators are boundedly similar on $L^2$ and on \(H^s(\partial D_P)\) and \(H^s(\partial D_Q)\), for every \(s\in\mathbb R\). No holomorphic or antiholomorphic M\"obius transformation relates their boundaries. Under the common rescaling $z\mapsto R^{-1/7}z$, both boundaries converge to the unit circle in every fixed $C^k$ norm as $R\to\infty$.
\end{theorem}

Figure~\ref{fig:counter} shows the pair at the threshold.
\begin{figure}[htbp]
\revisionfloat
\centering
\includegraphics[width=\textwidth]{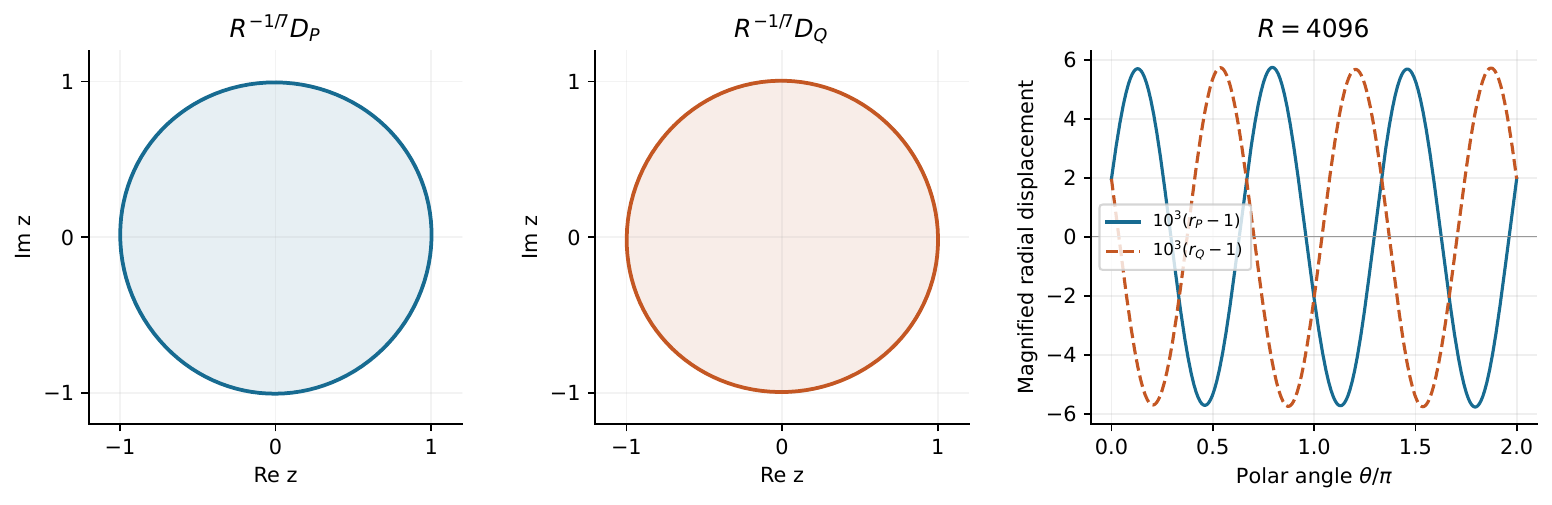}
\caption{The counterexample at $R=4096$, rescaled by $R^{-1/7}$. Left and middle: the domains and the unit circle (dashed grey). Right: $10^3(r_P(\theta)-1)$ and $10^3(r_Q(\theta)-1)$ where the normalised polar radii are computed from the polynomial equations.}\label{fig:counter}
\end{figure}

\begin{proof}
\emph{1. Boundary geometry.}
Let $p(z)=z^7+Az^4+Bz+i$ denote either polynomial in \eqref{eq:PQ}. Since $|A|<3/2$ and $|B|<1/2$, for $|z|\le3$,
\[
 |p(z)|\le3^7+\tfrac32 3^4+\tfrac12 3+1=2311<4096.
\]
For $|z|\ge3$, write
\[
 p=z^7(1+e),\quad p'=7z^6(1+f),\quad p''=42z^5(1+g),
\]
with
\[
 e=Az^{-3}+Bz^{-6}+iz^{-7},\quad
 f=\tfrac47Az^{-3}+\tfrac17Bz^{-6},\quad g=\tfrac27Az^{-3}.
\]
Then
\[
 |e|<\tfrac3{50},\qquad |f|<\tfrac1{30},\qquad |g|<\tfrac1{60}.
\]
Thus $p,p'$ have no zeros there, and
$\Re(zp'/p)>0$. Along each ray, $|p|$ increases strictly beyond radius three, from a value below $R$ to infinity. The level set is therefore a single analytic radial Jordan curve.

To check convexity, parametrise the curve locally by $p(z(t))=Re^{it}$. Its tangent angle satisfies
\[
 \frac{\dd}{\dd t}\arg z'(t)=\Re\left(1-\frac{pp''}{(p')^2}\right).
\]
The estimates above give
\[
 \left|\frac{(1+e)(1+g)}{(1+f)^2}-1\right|
 <\frac{1309}{8410}<\frac16,
\]
so
\[
 \frac{\dd}{\dd t}\arg z'(t)>
 \frac17-\frac67\frac{1309}{8410}
 =\frac{278}{29435}>0.
\]
The curvature is positive. The parametrisation is local because one circuit of the boundary maps to seven circuits of the target circle. Every target value $\tau$ with $|\tau|\ge4096$ is regular for both polynomials.

Rescale by $z=R^{1/7}u$, so that the boundary equation for $P$ becomes
\[
 |u^7+\alpha R^{-3/7}u^4+\beta R^{-6/7}u+iR^{-1}|=1,
\]
and the equation for $Q$ uses the conjugates of $\alpha$ and $\beta$. For the limiting equation $|u^7|=1$, the radial derivative on $|u|=1$ is nonzero. The implicit function theorem gives convergence of the polar radii to one in every fixed $C^k$ norm.

\medskip\noindent\emph{2. Similarity of the NP operators.}
The construction rests on the exact factorisation
\begin{equation}\label{eq:factor}
 P(x)-Q(y)=H(x,y)G(x,y),
\end{equation}
where
\begin{align*}
 H(x,y)&=x^3+\frac{1+d}{2}x^2y+\frac{-1+d}{2}xy^2-y^3-\frac d7,\\
 G(x,y)&=x^4-\frac{1+d}{2}x^3y-x^2y^2+\frac{-1+d}{2}xy^3+y^4
 -\frac{7+5d}{14}x+\frac{-7+5d}{14}y.
\end{align*}
Over a common regular value $t$, a root $x$ of $P(x)=t$ is joined to a root $y$ of $Q(y)=t$ if $H(x,y)=0$. The $7\times7$ incidence matrix $N$, with rows indexed by $y$, has three ones in each row and column. The three roots incident to any fixed root are distinct: at a zero of $H$, the factorisation gives $P'=H_xG$ and $-Q'=H_yG$, so regularity excludes $G=0$, $H_x=0$, and $H_y=0$.

Near infinity, label the roots so that $x_j\sim t^{1/7}\zeta^j$ and $y_k\sim t^{1/7}\zeta^k$, where $\zeta=e^{2\pi i/7}$. The identity
\[
 \zeta+\zeta^2+\zeta^4=\frac{-1+d}{2}
\]
gives
\[
 H_3(X,1)=\prod_{q\in\{3,5,6\}}(X-\zeta^q),
\]
where $H_3$ is the homogeneous cubic part of $H$. Thus, the incidences are $j-k\in\{3,5,6\}$ modulo seven. Their six ordered nonzero differences exhaust the nonzero residues modulo seven. Hence
\begin{equation}\label{eq:N}
 NN^{\mathsf T}=N^{\mathsf T}N=2I+J_7,
 \qquad N^{-1}=\tfrac12N^{\mathsf T}-\tfrac16J_7.
\end{equation}
Now continue the roots along a path avoiding the finitely many critical values of $P$ and $Q$. The complement is connected, and an incidence cannot appear or disappear: at every zero of $H$ over a regular value, the identities $P'=H_xG$ and $-Q'=H_yG$ imply $H_xH_yG\ne0$. Continuation changes only the row and column labels of $N$, so \eqref{eq:N} holds over every regular value, in particular on $|t|=R$. Here, $J_7$ is the all-ones matrix.

To pass from incidences to integral operators, we use
\begin{equation}\label{eq:logidentity}
 \sum_{H(x,y)=0}\frac1{P'(x)(x-x')}
 =\frac1{Q'(y)}\sum_{H(x',y')=0}\frac1{y-y'}.
\end{equation}
Locally, $H(X,y)=\prod_{j=1}^3(X-x_j(y))$. Since $P(x_j(y))=Q(y)$, differentiation gives $x_j'=Q'/P'(x_j)$. Computing $\partial_y\log H(x',y)$ gives the left side multiplied by $Q'(y)$. Factoring instead in $y$, with constant leading coefficient $-1$, gives the right sum. Both calculations hold away from the displayed poles.

Let $M_P$ and $M_Q$ denote multiplication by $|P'|$ and $|Q'|$ on the two boundaries. The map
\begin{equation}\label{eq:V}
 (V\varphi)(y)=|Q'(y)|\sum_{H(x,y)=0}\frac{\varphi(x)}{|P'(x)|}
\end{equation}
is $M_QNM_P^{-1}$, whose inverse is given by \eqref{eq:N}:
\[
(V^{-1}\psi)(x)=|P'(x)|
\left[
\frac12\sum_{H(x,y)=0}\frac{\psi(y)}{|Q'(y)|}
-\frac16\sum_{Q(y)=P(x)}\frac{\psi(y)}{|Q'(y)|}
\right].
\]
The coverings are smooth and unbranched. Thus both maps, finite sums of smooth compositions with nonvanishing multipliers, are bounded on every $H^s$, including $L^2$.

On the first boundary,
\[
 \frac{\nu_P(x)}{|P'(x)|}=\frac{P(x)}{RP'(x)}.
\]
Changing variables to $t=P(x)$ and $s=P(x')$ gives the conjugated kernel
\[
 \frac1{2\pi R}\Re\frac{t}{P'(x)(x-x')}
\]
together with the sum over $P(x')=s$ and integration against $|\dd s|$.
Multiply \eqref{eq:logidentity} by $t$, take real parts, and interchange the finite incidence sums. This gives
\begin{equation}\label{eq:intertwine}
 K_{D_Q}^*V=VK_{D_P}^*.
\end{equation}
Initially the target values $s,t$ are distinct. Continuity of the NP kernels across the diagonal extends the identity to the integral operators. A change of variables, using the column sums of $N$, also gives
$\int_{\partial D_Q}V\varphi\dd s=3\int_{\partial D_P}\varphi\dd s$.
Thus the similarity preserves the mean-zero spaces and gives equality of the full spectra, multiplicities and kernel dimensions.

\medskip\noindent\emph{3. M\"obius non-equivalence.}
For a polynomial $A$ of positive degree, let $A^\#(z)=\overline{A(\bar z)}$. We need a fact about polynomial lemniscates. Suppose $A,B$ have the same degree $n$, that $A$ has a simple zero, and that $A$ is not a constant plus a pure $n$th power. Then every holomorphic M\"obius map $M$ from the whole lemniscate $\{|A|=R\}$ onto $\{|B|=R\}$ is affine and satisfies $B\circ M=\tau A$ for some $|\tau|=1$.

If $A$ has a simple zero and $R>0$, then
\[
 A(z)A^\#(w)-R^2
\]
is irreducible in $\C[z,w]$. This is a special case of Orevkov's irreducibility criterion \cite{Orevkov}. Here Eisenstein's criterion gives a short proof. Reverse the polynomial in $z$: for $n=\deg A$, use $u^n[A(1/u)A^\#(w)-R^2]$. If $A(z)=\sum_{j=0}^n a_jz^j$ and $w_0$ is a simple zero of $A^\#$, the coefficient of $u^n$ is $a_0A^\#(w)-R^2$, which is not divisible by $w-w_0$. Every lower coefficient is divisible by $w-w_0$, and the constant coefficient $a_nA^\#(w)$ is not divisible by $(w-w_0)^2$. Eisenstein gives irreducibility in $\C(w)[u]$. The coefficients have no common factor, since $A^\#$ and $a_0A^\#-R^2$ are relatively prime. Gauss's lemma and reversal in $u$ now give irreducibility in $\C[z,w]$.

Let $M(z)=(\alpha_0z+\beta_0)/(\gamma_0z+\eta_0)$ be such a map, and set
\[
 D(z)=(\gamma_0z+\eta_0)^n,\qquad C(z)=D(z)B(M(z)).
\]
The polynomial $C(z)C^\#(w)-R^2D(z)D^\#(w)$ vanishes at infinitely many points of the complexified lemniscate. It is not identically zero: $C,D$ are linearly independent, and its coefficient matrix has rank two. Two plane algebraic curves without a common component have only finitely many intersection points. Thus the irreducible defining polynomial divides this numerator. The numerator has degree at most $n$ in each variable, while the defining polynomial has degree exactly $n$ in each. Their quotient is a nonzero constant. The column space of the coefficient matrix of $C(z)C^\#(w)-R^2D(z)D^\#(w)$ is $\operatorname{span}\{C,D\}$; that for $A(z)A^\#(w)-R^2$ is $\operatorname{span}\{A,1\}$. Hence,
\[
 \operatorname{span}\{C,D\}=\operatorname{span}\{A,1\}.
\]
If $\gamma_0\ne0$, this would express $A$ as a constant plus a pure $n$th power. Thus, $M$ is affine. The span identity gives $B(M(z))=c_1A(z)+c_0$, with $c_1\ne0$. For every $\tau$ on $|\tau|=R$, a root of $A(z)=\tau$ belongs to its lemniscate, so $|c_1\tau+c_0|=R$. Squaring and comparing the constant and first Fourier coefficients gives $c_0=0$ and $|c_1|=1$. Therefore,
\begin{equation}\label{eq:affine}
 B(\alpha_0z+\beta_0)=\tau A(z),\qquad |\tau|=1.
\end{equation}

We check these hypotheses for the polynomials in \eqref{eq:PQ}. The equations $P'(z)=0$ and $Q'(z)=0$ are quadratic in $z^3$, with nonzero constant terms and respective discriminants $-14+6d$ and $-14-6d$, both nonzero. Each polynomial therefore has six distinct critical points. Each degree-seven polynomial has a simple zero, since its multiple zeros have multiplicity at most two and seven is odd. Neither $P$ nor $P^\#$ is a constant plus a pure seventh power, since its sixth-degree coefficient vanishes and its fourth-degree coefficient does not. The argument applies to $P$ and $P^\#$.

A holomorphic equivalence gives $Q(\alpha_0z+\beta_0)=\tau P(z)$. Comparing degree-six and constant terms gives $\beta_0=0$ and $\tau=1$. The other coefficients would require
\[
 \alpha_0^7=1,\quad \bar\alpha\alpha_0^4=\alpha,\quad
 \bar\beta\alpha_0=\beta,
 \quad\text{hence}\quad \frac\beta{\alpha^2}=\frac{\bar\beta}{\bar\alpha^2}.
\]
The two ratios, however, are $(11+d)/56$ and $(11-d)/56$. For an antiholomorphic equivalence, apply \eqref{eq:affine} to $P^\#(z)=z^7+\bar\alpha z^4+\bar\beta z-i$. The constants force $\tau=-1$, the linear term forces $\alpha_0=-1$, and the fourth-degree term then gives $\bar\alpha=-\bar\alpha$, again a contradiction.
\end{proof}

The counterexample leaves room for uniqueness within prescribed families. We now obtain it by making the higher conformal coefficients depend on finitely many free parameters.

\section{Finite-dimensional uniqueness and stability}\label{sec:class}
\subsection{The family and the inverse theorem}\label{sec:main}
To obtain stable recovery despite nonuniqueness in the unrestricted problem, we prescribe a finite-dimensional family. We use exterior conformal maps with coefficient of $w$ equal to one, constant coefficient zero, and real Laurent coefficients in fixed coordinates. This normalisation fixes capacity and translation; orientation is also prescribed, as the NP spectrum is invariant under these geometric changes. We seek to recover the parameters of the family below from its reduced spectrum.

Fix $m\in\N$ and a known tail
\begin{equation}\label{eq:tail}
 F(w)=\sum_{n\ge m+1}f_nw^{-n},\qquad f_n\in\R,
 \qquad |f_n|\le C_F r_*^n\quad(0<r_*<1).
\end{equation}
For real $\avec=(a_1,\ldots,a_m)$, set
\begin{equation}\label{eq:family}
 \Psi_{\avec}(w)=w+\sum_{n=1}^m a_nw^{-n}+|\avec|^2F(w),
 \qquad |\avec|^2=\sum_{n=1}^m a_n^2.
\end{equation}
Let $D_{\avec}$ be the bounded region enclosed by $\Psi_{\avec}(\{|w|=1\})$, and set $K_{\avec}=\overline{D_{\avec}}$. Since
\( \Psi_{\avec}(\overline{w}) = \overline{\Psi_{\avec}(w)}\), the domain $D_{\avec}$ is symmetric about the real axis.

For small parameters, these maps define domains, as proved below. The first $m$ Laurent coefficients vary independently, while \eqref{eq:family} prescribes all higher coefficients. The choice $F=0$ is permitted.

For a unit vector $v\in\R^m$ and $\varepsilon,\delta>0$, define
\begin{equation}\label{eq:cone}
 \mathcal C(v,\varepsilon,\delta)
 =\{tv+u:0<t<\delta,\ u\perp v,\ |u|<\varepsilon t\}.
\end{equation}
This truncated cone is convex and $m$-dimensional, with the usual interpretation when $m=1$. Every point satisfies $t\le|\avec|\le\sqrt{1+\varepsilon^2}\,t$. The stability estimates extend by continuity to the vertex $\avec=0$; spectral derivatives are taken only on the punctured cone.

Write $\lambda_1(\avec)\ge\lambda_2(\avec)\ge\cdots\ge0$ for the positive reduced NP eigenvalues, counted with multiplicity and padded with zeros if there are only finitely many. At the disk the entire sequence is zero. Define the leading spectral map by
\[
 \Lambda_m(\avec)=(\lambda_1(\avec),\ldots,\lambda_m(\avec)).
\]

The linear part of the Grunsky operator is determined by the real symmetric matrix
\begin{equation}\label{eq:L}
 (L_m(\avec))_{ij}=
 \begin{cases}\sqrt{ij}\,a_{i+j-1},&i+j\le m+1,\\0,&i+j>m+1,
 \end{cases}\qquad 1\le i,j\le m.
\end{equation}
Define its characteristic coefficients $e_k^{(m)}(\avec)$ by
\begin{equation}\label{eq:E}
 \det(zI-L_m)=\sum_{k=0}^m(-1)^ke_k^{(m)}z^{m-k},\qquad
 e_0^{(m)}=1,\qquad \mathcal E_m=(e_1^{(m)},\ldots,e_m^{(m)}).
\end{equation}

\begin{definition}\label{def:good}
A unit vector $v$ is admissible for the spectral argument if
\begin{enumerate}
\item[(i)] $\det D\mathcal E_m(v)\ne0$;
\item[(ii)] all eigenvalues $\mu_j(v)$ of $L_m(v)$ are nonzero;
\item[(iii)] their absolute values are pairwise distinct.
\end{enumerate}
The third condition makes the signed eigenvalues simple.
\end{definition}

\begin{theorem}[Finite-parameter spectral uniqueness and stability]\label{thm:main}
For every $m\ge1$ there is an admissible unit vector $v$ with positive coordinates. Fix any admissible $v$ and any tail satisfying \eqref{eq:tail}. There exist $\varepsilon,\delta,c,C>0$ for which the following conclusions hold on $\mathcal C=\mathcal C(v,\varepsilon,\delta)$.
\begin{enumerate}
\item[(i)] Each $\Psi_{\avec}$ is the exterior conformal map of a strictly convex analytic Jordan domain. Its first $m$ positive reduced NP eigenvalues are simple and separated from all remaining positive eigenvalues;
\item[(ii)] For all $\avec,\bvec\in\mathcal C\cup\{0\}$,
\begin{equation}\label{eq:spectralbilip}
 c|\avec-\bvec|\le
 |\Lambda_m(\avec)-\Lambda_m(\bvec)|
 \le C|\avec-\bvec|.
\end{equation}
Thus these $m$ eigenvalues uniquely determine the normalised domain within this class;
\item[(iii)] For all $\avec,\bvec\in\mathcal C\cup\{0\}$, let $c(\avec)=(c_n(\avec))_{n\ge1}$ denote the complete Laurent coefficient sequence in \eqref{eq:family}. Then, for each fixed $s\ge0$,
\begin{equation}\label{eq:consequences}
 d_H(K_{\avec},K_{\bvec})+
 \norm{c(\avec)-c(\bvec)}_{\ell_s^2}
 \le C_s|\Lambda_m(\avec)-\Lambda_m(\bvec)|,
\end{equation}
where $\norm{c}_{\ell_s^2}^2=\sum_{n\ge1}n^{2s}|c_n|^2$.
\end{enumerate}
The constants may depend on $m,F,v$ and the cone aperture. With these data fixed, they remain uniform as $\avec,\bvec\to0$ within the cone. No bound on their dependence on $m$ is asserted.
\end{theorem}

The proof begins with geometric coordinates and estimates for the full Grunsky operator. We then construct admissible directions for its finite-dimensional linearisation and compare the derivatives of the two spectral maps uniformly near the disk. Integrating this comparison along segments in a convex cone gives \eqref{eq:spectralbilip}. Uniform control of the differential throughout the cone is essential to this final step.

\subsection{Conformal geometry and geometric coordinates}\label{sec:geometry}
Define
\[
 T_n(w,z)=\sum_{k=1}^n w^{-k}z^{-(n+1-k)}.
\]
The divided difference of \eqref{eq:family} is
\begin{equation}\label{eq:divided}
 \frac{\Psi_{\avec}(w)-\Psi_{\avec}(z)}{w-z}=1-\mathcal S_{\avec}(w,z),
 \quad
 \mathcal S_{\avec}=\sum_{n=1}^m a_nT_n+|\avec|^2\sum_{n>m}f_nT_n.
\end{equation}
Since $|T_n(w,z)|\le n$ for $|w|,|z|\ge1$, the condition
\begin{equation}\label{eq:univalent}
 \sum_{n=1}^m n|a_n|+|\avec|^2\sum_{n>m}n|f_n|<1
\end{equation}
ensures injectivity on the closed exterior disk. Taking $z\to w$ also gives $\Psi'_{\avec}(w)\ne0$ there. The normalisation $\Psi_{\avec}(w)=w+O(w^{-1})$ at infinity then shows that the exterior disk maps onto the unbounded component of the complement of the boundary Jordan curve. Hence $\Psi_{\avec}$ is the normalised exterior conformal map of that curve.

Exponential decay of the tail makes $z_{\avec}(\theta)=\Psi_{\avec}(e^{i\theta})$ analytic in $\theta$ and gives $z_{\avec}\to e^{i\theta}$ in $C^k$ for every fixed $k$. In particular,
\[
 \frac{\operatorname{Im}(\overline{z'_{\avec}}z''_{\avec})}
 {|z'_{\avec}|^3}\longrightarrow1
\]
uniformly in $\theta$. Thus the curvature is positive for small parameters, and the simple closed curve bounds a strictly convex body. These conclusions hold throughout a sufficiently small real parameter ball centred at the origin.

Near the disk, the Laurent parameters give quantitative geometric coordinates. To compare parameter distance with Hausdorff distance on a small real parameter ball, we differentiate the support function.

\begin{proposition}[Uniform geometric coordinates]\label{prop:geometry}
There exist $r_0>0$ and $C_+>0$ such that, whenever $|\avec|,|\bvec|\le r_0$,
\begin{equation}\label{eq:hausbilip}
 \frac1{2\sqrt2}|\avec-\bvec|
 \le d_H(K_{\avec},K_{\bvec})\le C_+|\avec-\bvec|.
\end{equation}
For every fixed $s\ge0$, define $F_s^2=\sum_{n>m}n^{2s}|f_n|^2$. Then
\begin{equation}\label{eq:coeffbound}
 \norm{c(\avec)-c(\bvec)}_{\ell_s^2}
 \le\sqrt{m^{2s}+4r_0^2F_s^2}\,|\avec-\bvec|.
\end{equation}
\end{proposition}
\begin{proof}
For a compact convex body $K$, define the support function
\[
 h_K(\theta)=\max_{z\in K}\Re(ze^{-i\theta}).
\]
For convex bodies, $d_H(K,L)=\norm{h_K-h_L}_\infty$. Indeed, $K\subset L+d\overline{\mathbb D}$ implies $h_K\le h_L+d$. Conversely, suppose this support inequality holds but $x\in K$ lies outside $L+d\overline{\mathbb D}$. The nearest point $y$ in that compact convex set satisfies $(x-y)\cdot(z-y)\le0$ for every $z$ in the set. The direction $x-y$ therefore strictly separates $x$ from the set, contradicting the support inequality. Interchanging $K$ and $L$ proves the identity.

We have $h_{\avec}(\theta)=\max_t\Re(\Psi_{\avec}(e^{it})e^{-i\theta})$. At $\avec=0$, the maximum occurs at $t=\theta$, where the second derivative in $t$ is $-1$. Outside a fixed small arc around this point, the disk objective is uniformly below its maximum. For small parameters, the maximiser stays in this arc and the second derivative stays uniformly negative there. The implicit function theorem gives a smooth maximiser; compactness of the angle circle makes its parameter derivative bounds uniform.

The first derivative in $t$ vanishes at the maximiser, so differentiation in the parameters gives
\[
 h_0=1,\qquad \partial_{a_n}h_0(\theta)=\cos((n+1)\theta).
\]
It follows that
\[
 h_{\avec}=1+\sum_{n=1}^m a_n\cos((n+1)\theta)+\mathcal R_{\avec},
 \quad
 \norm{\mathcal R_{\avec}-\mathcal R_{\bvec}}_\infty
 \le Mr_0|\avec-\bvec|.
\]
For the remainder bound, bounded second parameter derivatives imply $\norm{D\mathcal R_{\avec}}\le M|\avec|$. Integrating along the segment between the parameters gives the stated difference estimate.

Orthogonality in $L^2(\dd\theta/(2\pi))$ gives
\[
 \left\|\sum_{n=1}^m(a_n-b_n)\cos((n+1)\theta)\right\|_{L^2}
 =\frac{|\avec-\bvec|}{\sqrt2}.
\]
The uniform norm dominates this normalised $L^2$ norm, so
\[
 (1/\sqrt2-Mr_0)|\avec-\bvec|
 \le\norm{h_{\avec}-h_{\bvec}}_\infty
 \le(\sqrt m+Mr_0)|\avec-\bvec|.
\]
Taking $Mr_0\le1/(2\sqrt2)$ proves \eqref{eq:hausbilip}.

The coefficient difference satisfies the exact identity
\[
 \norm{c(\avec)-c(\bvec)}_{\ell_s^2}^2
 =\sum_{n=1}^m n^{2s}|a_n-b_n|^2
 +(|\avec|^2-|\bvec|^2)^2F_s^2.
\]
Using $\big||\avec|^2-|\bvec|^2\big|\le2r_0|\avec-\bvec|$ and $n^{2s}\le m^{2s}$ for $n\le m$ proves \eqref{eq:coeffbound}. Exponential decay gives $F_s<\infty$ for every fixed $s$.
\end{proof}

Boundary points with the same conformal angle are at distance at most $C|\avec-\bvec|$, which gives the same upper bound for $d_H(\partial D_{\avec},\partial D_{\bvec})$. The support-function argument bounds body Hausdorff distance by boundary Hausdorff distance. Together with the geometric coordinate estimate, these bounds make the two Hausdorff distances comparable on the small parameter ball.

\subsection{The Grunsky operator and the NP spectrum}\label{sec:analytic}
We use the convention
\begin{equation}\label{eq:grunsky}
 \log\frac{\Psi_{\avec}(w)-\Psi_{\avec}(z)}{w-z}
 =-\sum_{i,j\ge1}b_{ij}(\avec)w^{-i}z^{-j},
 \qquad B_{ij}(\avec)=\sqrt{ij}\,b_{ij}(\avec).
\end{equation}
Symmetry gives $b_{ij}=b_{ji}$. Throughout, the coefficients in \eqref{eq:grunsky} are logarithmic Grunsky coefficients; the corresponding Faber-polynomial coefficients carry an additional index factor.

\begin{lemma}[Analytic operator expansion]\label{lem:analytic}
On a sufficiently small parameter ball, $B(\avec)$ is a real symmetric Hilbert--Schmidt operator on $\ell^2(\N)$, depends real analytically on $\avec$, and satisfies
\begin{equation}\label{eq:remainder}
 B(\avec)=\widetilde L_m(\avec)+R(\avec),\qquad
 \norm{R(\avec)}_{\HS}\le C|\avec|^2,\qquad
 \norm{DR(\avec)}_{\R^m\to\HS}\le C|\avec|,
\end{equation}
where $\widetilde L_m=L_m\oplus0$.
\end{lemma}
\begin{proof}
Choose $r_*<r<1$ and complexify the parameters near the origin, replacing $|\avec|^2$ by the holomorphic polynomial $\sum a_n^2$. For $|w|,|z|\ge r$,
\[
 |T_n(w,z)|\le nr^{-n-1},\qquad
 \sum_{n>m}n|f_n|r^{-n-1}<\infty.
\]
On a sufficiently small complex parameter ball, we have $|\mathcal S_{\avec}|<q<1$ uniformly in this region, so
\begin{equation}\label{eq:logseries}
 -\log(1-\mathcal S_{\avec})
 =\sum_{k\ge1}\frac{\mathcal S_{\avec}^k}{k}
\end{equation}
converges normally. Cauchy estimates in $w^{-1},z^{-1}$ give
\[
 |b_{ij}(\avec)|\le C r^{i+j}.
\]
The convergence of $\sum_{i,j}ijr^{2i+2j}$ shows that these coefficients define a holomorphic map into the Hilbert--Schmidt operators. Parameter Cauchy estimates give the corresponding bounds for its derivatives on a smaller ball.

The degree-one term of \eqref{eq:logseries} is $\sum_{n\le m}a_nT_n$. Its coefficient of $w^{-i}z^{-j}$ equals $a_{i+j-1}$ when $i+j\le m+1$ and is zero otherwise. Multiplying by $\sqrt{ij}$ gives \eqref{eq:L}. All remaining terms have parameter degree at least two. Taylor's theorem, with a bound on the second derivative of the Hilbert--Schmidt-valued map, gives both estimates in \eqref{eq:remainder}. For real parameters, the operator is real and symmetric, hence self-adjoint.
\end{proof}

For real coefficients, the NP--Grunsky representation gives the bounded similarity
\begin{equation}\label{eq:block}
 K_{D_{\avec}}^*\big|_{H^{-1/2}_0} \sim\
 \frac12\begin{pmatrix}0&B(\avec)\\B(\avec)&0\end{pmatrix}.
\end{equation}
To describe the conformal density coordinates used here, set
\[
 \zeta_{\pm n}(\Psi(e^{i\theta}))
 =\frac{\sqrt n\,e^{\pm in\theta}}{h(\theta)},\qquad
 h(\theta)=|\partial_\theta\Psi(e^{i\theta})|.
\]
Expansion in these densities gives a bounded identification of $\ell^2\oplus\ell^2$ with the mean-zero $H^{-1/2}$ space. The Faber expansion and layer-potential jump relations give
\[
 K_D^*\zeta_n=\frac12\sum_{k\ge1}B_{kn}\zeta_{-k},\qquad
 K_D^*\zeta_{-n}=\frac12\sum_{k\ge1}B_{kn}\zeta_k.
\]
This proves \eqref{eq:block}; see \cite[Proposition 3.1 and equation (4.3)]{JLdecay} and the derivation in \cite{JLseries}. The classical relation with Fredholm eigenvalues is discussed in \cite{Schiffer}.

The unitary change $(x,y)\mapsto2^{-1/2}(x+y,x-y)$ transforms the block operator into $\frac12\diag(B,-B)$. Its positive eigenvalues, with multiplicity, are $\frac12$ times the nonzero absolute eigenvalues of $B$. Self-adjointness identifies these absolute eigenvalues with the singular values, so
\begin{equation}\label{eq:singular}
 s_j(B(\avec))=2\lambda_j(\avec).
\end{equation}
The real Laurent coefficients ensure that $B$ is self-adjoint and allow the signed-eigenvalue perturbation argument below.

\subsection{Nondegenerate directions}\label{sec:directions}
\begin{lemma}[A determinant identity]\label{lem:det}
For every $m$,
\begin{equation}\label{eq:det}
 \det L_m(\avec)=\kappa_m a_m^m,\qquad
 \kappa_m=(-1)^{m(m-1)/2}m!.
\end{equation}
\end{lemma}
\begin{proof}
Let $J$ reverse the $m$ columns. Then $(L_mJ)_{ik}=0$ for $i>k$, and the diagonal entries are $\sqrt{i(m+1-i)}a_m$. Their product is $m!a_m^m$. Since $\det J=(-1)^{m(m-1)/2}$, the determinant of $L_m$ is given by \eqref{eq:det}.
\end{proof}

\begin{proposition}[Existence in the positive orthant]\label{prop:existence}
For every $m\ge1$, there is a vector $p\in(0,\infty)^m$ for which $D\mathcal E_m(p)$ is invertible and $L_m(p)$ has nonzero eigenvalues with pairwise distinct absolute values. The normalised vector $v=p/|p|$ is therefore admissible.
\end{proposition}
\begin{proof}
We induct on $m$. For $m=1$, $L_1(a_1)=(a_1)$, so any $a_1>0$ has the required properties. Suppose $p'\in\R^{m-1}$ is positive and satisfies the induction hypothesis. Setting $a_m=0$ gives
\[
 L_m(p',0)=L_{m-1}(p')\oplus0,
 \qquad e_k^{(m)}(p',0)=e_k^{(m-1)}(p')\quad(1\le k<m).
\]
The last characteristic coefficient is the determinant in \eqref{eq:det}. Expanding the Jacobian determinant along its last row gives
\begin{equation}\label{eq:jacinduction}
 \det D\mathcal E_m(\avec)
 =m\kappa_m a_m^{m-1}
 \det\left[\frac{\partial e_i^{(m)}}{\partial a_j}(\avec)\right]_{i,j=1}^{m-1}.
\end{equation}
At $(p',0)$, the minor equals $\det D\mathcal E_{m-1}(p')\ne0$. It remains nonzero for sufficiently small positive $a_m$, by continuity, and the prefactor in \eqref{eq:jacinduction} is then nonzero as well.

The first $m-1$ eigenvalues stay close to their simple nonzero limits and retain distinct absolute values. The last tends to zero as $a_m\to0$, but \eqref{eq:det} makes it nonzero for $a_m>0$. Taking $a_m$ smaller if necessary makes its magnitude smaller than that of every other eigenvalue, completing the induction. Since $e_k^{(m)}$ is homogeneous of degree $k$, $\det D\mathcal E_m(tp)=t^{m(m-1)/2}\det D\mathcal E_m(p)$ for $t>0$. Positive scaling also preserves the spectral conditions, so $v=p/|p|$ is admissible.
\end{proof}

The construction takes each new coefficient sufficiently small relative to the preceding ones. Admissibility is also open and dense in the positive orthant. The Jacobian is a nonzero polynomial. Nonzero eigenvalues with pairwise distinct absolute values correspond to positive simple eigenvalues of $L_m^2$; the excluded cases are the zeros of its determinant and discriminant. Proposition~\ref{prop:existence} shows that neither polynomial vanishes identically. Each admissible direction has its own cone and constants.

\begin{lemma}[Invertible singular-value differential]\label{lem:ds}
At an admissible $v$, the ordered singular-value map
\[
 S_0(\avec)=(s_1^0(\avec),\ldots,s_m^0(\avec)),\qquad
 s_j^0(\avec)=s_j(L_m(\avec))
\]
is smooth near $v$, with invertible differential $A=DS_0(v)$. On a narrow cone about $v$, $S_0$ is homogeneous of degree one and $DS_0$ of degree zero.
\end{lemma}
\begin{proof}
Choose smooth local labels $\mu_j(\avec)$ for the simple signed eigenvalues. The Jacobian determinant of their elementary symmetric functions is, up to sign, $\prod_{i<j}(\mu_i-\mu_j)$. To verify that it is nonzero, differentiate $\prod_j(z-\mu_j)$ with respect to the roots. If a linear combination of the resulting degree-$(m-1)$ polynomials vanishes, evaluation at each distinct root makes its corresponding coefficient zero. Since $D\mathcal E_m(v)$ is invertible, the chain rule implies that $D(\mu_1,\ldots,\mu_m)(v)$ is invertible.

Near $v$, the eigenvalues stay nonzero with distinct absolute values. Passing to ordered singular values is therefore a fixed signed permutation and preserves invertibility of the differential. Linearity of $L_m$ gives $S_0(t\avec)=tS_0(\avec)$ for $t>0$; differentiating gives $DS_0(t\avec)=DS_0(\avec)$.
\end{proof}

\subsection{Differentiated spectral estimates}\label{sec:perturbation}
Near the disk, the linearised eigenvalue gaps decay linearly in the parameters, while the Grunsky remainder is quadratic. This difference of orders also controls derivatives of the leading spectral branches. Analytic perturbation theory supplies smooth local eigenvectors for isolated simple eigenvalues; see \cite[Chapter VII, Section 1]{Kato}. We now establish estimates that remain uniform as the parameters tend to zero within the cone.

Fix an admissible $v$ and a small closed angular neighbourhood on which the conditions of Definition~\ref{def:good} persist. Compactness and homogeneity give $\gamma>0$ such that, on the corresponding cone, with $r=|\avec|$,
\begin{equation}\label{eq:gaps}
 |\mu_j|\ge\gamma r,\qquad
 |\mu_i-\mu_j|\ge\gamma r,\qquad
 \big||\mu_i|-|\mu_j|\big|\ge\gamma r\quad(i\ne j).
\end{equation}
Here $\mu_j$ are the nonzero eigenvalues of $\widetilde L_m$, whose remaining spectral value is zero.

\begin{lemma}[Eigenvalue and derivative estimates]\label{lem:perturb}
On a sufficiently small such cone, the first $m$ singular values of $B$ are simple. They are the absolute values of eigenvalues $\beta_j$ perturbing the $m$ nonzero eigenvalues $\mu_j$ of $\widetilde L_m$. The map
\[
 S(\avec)=(s_1(B(\avec)),\ldots,s_m(B(\avec)))
\]
satisfies the uniform estimates
\begin{equation}\label{eq:Serror}
 |S(\avec)-S_0(\avec)|\le Cr^2,\qquad
 \norm{DS(\avec)-DS_0(\avec)}\le Cr.
\end{equation}
The remaining singular values obey $s_{m+1}(B)\le Cr^2$.
\end{lemma}
\begin{proof}
Consider the circles $|z-\mu_j|=\gamma r/4$. Their distance from $\sigma(\widetilde L_m)$ is at least $\gamma r/4$. Shrink the parameter radius so that $\|R\|_{\op}\le\gamma r/8$. The resolvent Neumann series converges on every circle along the path $\widetilde L_m+\tau R$, $0\le\tau\le1$, so the corresponding Riesz projections keep rank one. Each circle thus encloses a unique simple eigenvalue $\beta_j$. The self-adjoint perturbation bound gives $|\beta_j-\mu_j|\le\norm{R}_{\op}\le Cr^2$ and bounds the displacement of the remaining spectrum by $\norm{R}_{\op}$.

Since $\widetilde L_m$ has rank $m$, the min--max characterisation gives $s_{m+1}(B)\le\norm{R}_{\op}\le Cr^2$. Each of the $m$ branches constructed above has magnitude at least $\gamma r-Cr^2$. Their absolute values are therefore exactly the largest $m$ singular values. The gaps in \eqref{eq:gaps} preserve the signs and absolute ordering of the eigenvalue branches.

To estimate derivatives, choose unit eigenvectors $u_j$ and $v_j$ of $B$ and $\widetilde L_m$ for the corresponding eigenvalues, and set $P_j=\langle\,\cdot\,,v_j\rangle v_j$. Project
\[
 (\widetilde L_m-\beta_j)u_j=-Ru_j
\]
onto $v_j^\perp$. By \eqref{eq:gaps}, the restriction of $\widetilde L_m-\beta_j$ to this subspace is invertible with inverse norm at most $C/r$. Hence
\[
 \norm{(I-P_j)u_j}\le Cr^{-1}\norm{R}_{\op}\le Cr.
\]
Choose the phase so that $\langle u_j,v_j\rangle\ge0$. As both vectors have unit norm, this gives $\norm{u_j-v_j}\le Cr$.

For each coordinate $a_k$, the simple-eigenvalue derivative formula gives
\[
 \partial_{a_k}\beta_j=\langle u_j,(\partial_{a_k}B)u_j\rangle,
 \qquad
 \partial_{a_k}\mu_j=\langle v_j,(\partial_{a_k}\widetilde L_m)v_j\rangle.
\]
Subtract these identities and apply \eqref{eq:remainder}. The matrix $\partial_{a_k}\widetilde L_m$ is fixed, bounded and finite-rank, so
\[
 |\partial_{a_k}\beta_j-\partial_{a_k}\mu_j|
 \le\norm{\partial_{a_k}R}_{\op}
 +2\norm{\partial_{a_k}\widetilde L_m}_{\op}\norm{u_j-v_j}
 \le Cr.
\]
The signs and ordering are fixed, so the same bounds hold for the singular values. Combining the finitely many branch and coordinate estimates proves the derivative estimate in \eqref{eq:Serror}.
\end{proof}

\subsection{Proof of the inverse theorem}\label{sec:inverse}
Integrating the differential comparison along segments in the convex cone now gives the inverse Lipschitz estimate.
\begin{proof}[Proof of Theorem~\ref{thm:main}]
Proposition~\ref{prop:existence} gives a positive admissible direction in every dimension. Fix any admissible $v$, and set $A=DS_0(v)$ and $\alpha=\sigma_{\min}(A)>0$. By Lemma~\ref{lem:ds}, $DS_0$ is continuous in direction and homogeneous of degree zero. Choose the cone aperture so that
\[
 \norm{DS_0(\avec)-A}\le\alpha/4
\]
on the punctured cone. Then choose its length $\delta$ small enough for the geometric conclusions of Subsection~\ref{sec:geometry}, the spectral conclusions of Lemma~\ref{lem:perturb}, and the bound
\[
 \norm{DS(\avec)-DS_0(\avec)}\le\alpha/4.
\]
Thus $\norm{DS(\avec)-A}\le\alpha/2$ throughout the cone. By convexity, the segment $\bvec+\tau(\avec-\bvec)$ joining any two points of the cone stays in the cone and avoids zero. Integration along this segment gives
\[
 S(\avec)-S(\bvec)=A(\avec-\bvec)+e,\qquad
 |e|\le\frac\alpha2|\avec-\bvec|.
\]
Therefore
\begin{equation}\label{eq:globalS}
 \frac\alpha2|\avec-\bvec|
 \le|S(\avec)-S(\bvec)|
 \le(\norm A+\alpha/2)|\avec-\bvec|.
\end{equation}
Equation \eqref{eq:singular} yields \eqref{eq:spectralbilip} with lower constant $c=\alpha/4$. As $S$ and $S_0$ tend to zero at the origin, continuity extends the estimate to the vertex. The geometric argument and Lemma~\ref{lem:perturb} prove part (i). Equality of the first $m$ eigenvalues forces equality of the parameters, hence of the normalised maps and domains.
Finally, combining the geometric and coefficient bounds in Proposition~\ref{prop:geometry} with the inverse bound in \eqref{eq:spectralbilip} proves \eqref{eq:consequences} for each fixed $s\ge0$.
\end{proof}

\begin{corollary}[Prescribed analytic remainders]\label{cor:generaltail}
Fix $m\ge1$ and $0<r_*<r<1$. For real parameters, consider
\[
 \Psi_{\avec}(w)=w+\sum_{n=1}^m a_nw^{-n}
 +\mathcal T(\avec;w),\qquad
 \mathcal T(\avec;w)=\sum_{n>m}r_n(\avec)w^{-n},
\]
where, for some $\rho>0$, the prescribed remainder is jointly holomorphic for complex parameters $|\avec|<\rho$ and $|w|>r_*$. Assume its Laurent coefficients are real for real parameters and $\sup_{|w|=r}|\mathcal T(\avec;w)|\le C_T|\avec|^2$ for all $|\avec|<\rho$. Then the conclusions of Theorem~\ref{thm:main} hold about every admissible unit vector for $L_m$, after choosing a suitable cone aperture and radius. The complete coefficient sequence is $c_n(\avec)=a_n$ for $n\le m$ and $c_n(\avec)=r_n(\avec)$ for $n>m$. For each fixed $s\ge0$, the constants remain uniform down to the disk and may depend on $m,s$, the direction, the prescribed remainder and the cone.
\end{corollary}
\begin{proof}
Choose $r<r_1<1$. On a smaller parameter ball, Cauchy estimates in $w$ and the parameters give
\[
 |r_n(\avec)|\le C|\avec|^2r^n,\qquad
 \norm{Dr_n(\avec)}_{\op}\le C|\avec|r^n.
\]
The tail contribution to the divided difference is therefore $O(|\avec|^2)$, and its parameter derivative is $O(|\avec|)$, uniformly for $|w|,|z|\ge r_1$. The logarithmic expansion then proves Lemma~\ref{lem:analytic} with the same linear part $\widetilde L_m$.

These coefficient bounds also give analytic, strictly convex boundaries on a small real parameter ball. The remainder and its first parameter derivative vanish at zero, so the support-function argument in Proposition~\ref{prop:geometry} has the same linearisation. For $|\avec|,|\bvec|\le r_0$,
\[
 \left(\sum_{n>m}n^{2s}|r_n(\avec)-r_n(\bvec)|^2\right)^{1/2}
 \le C_s r_0|\avec-\bvec|,
\]
by integration along the parameter segment and convergence of $\sum n^{2s}r^{2n}$. The geometric and weighted coefficient bounds used in the main proof therefore remain valid. Admissibility depends only on $L_m$, and the differentiated spectral argument requires exactly these remainder estimates. The proof of Theorem~\ref{thm:main} now gives all three conclusions.
\end{proof}

\begin{remark}[What provides global uniqueness]
The estimate $DS\approx A$ compares the differential everywhere in the cone with a single invertible matrix. Convexity turns this uniform comparison into a lower bound for every secant; pointwise invertibility alone does not give this conclusion. For $m=2$, the explicit inverse of the leading spectral map yields the larger sector in Subsection~\ref{sec:m2}.
\end{remark}

\subsection{Complete spectra and further stability estimates}\label{sec:spectralmetrics}\label{sec:forward}\label{sec:stability}
For the complete ordered positive spectrum, define
\[
 d_{\mathrm{spec},2}(\avec,\bvec)
 =\left(\sum_{j\ge1}|\lambda_j(\avec)-\lambda_j(\bvec)|^2\right)^{1/2},
 \quad
 d_{\mathrm{spec},\infty}(\avec,\bvec)
 =\sup_{j\ge1}|\lambda_j(\avec)-\lambda_j(\bvec)|.
\]
Lemma~\ref{lem:analytic} and \eqref{eq:singular} imply that $d_{\mathrm{spec},2}$ is finite. Since $|\Delta\Lambda_m|\le d_{\mathrm{spec},2}$ and $|\Delta\Lambda_m|\le\sqrt m\,d_{\mathrm{spec},\infty}$, the coefficient and geometric estimates of Theorem~\ref{thm:main} hold with either full spectral distance on the right. Both distances retain eigenvalue ordering and multiplicity.

A uniform Hilbert--Schmidt bound on the derivative of the Grunsky operator gives forward spectral stability throughout a sufficiently small parameter ball.

\begin{lemma}[Singular-value perturbation in Hilbert--Schmidt norm]\label{lem:mirsky}
For Hilbert--Schmidt operators $A,C$ on the same separable Hilbert space, with ordered singular-value sequences padded by zeros,
\[
 \left(\sum_{j\ge1}|s_j(A)-s_j(C)|^2\right)^{1/2}
 \le\norm{A-C}_{\HS}.
\]
\end{lemma}
\begin{proof}
In finite dimensions, the trace inequality
$\Re\Tr(A^*C)\le\sum_j s_j(A)s_j(C)$ gives
\[
 \sum_j(s_j(A)-s_j(C))^2
 \le\norm A_{\HS}^2+\norm C_{\HS}^2-2\Re\Tr(A^*C)
 =\norm{A-C}_{\HS}^2.
\]
This is the Hilbert--Schmidt case of the singular-value inequalities in \cite{Mirsky}. In infinite dimensions, choose increasing finite-rank orthogonal projections $P_N$ converging strongly to the identity. Then $P_NAP_N\to A$ and $P_NCP_N\to C$ in Hilbert--Schmidt norm. By the min--max characterisation, each singular value changes by at most the operator norm of the perturbation, so every fixed singular value converges. Apply the finite-dimensional inequality to the compressions, keep the first $J$ summands and let $N\to\infty$. Finally, let $J\to\infty$.
\end{proof}

\begin{theorem}[Forward stability of the complete NP spectrum]\label{thm:forward}
On a sufficiently small real parameter ball in the fixed family, there is a constant $C$ such that
\begin{equation}\label{eq:forward}
 d_{\mathrm{spec},2}(\avec,\bvec)
 \le\tfrac12\norm{B(\avec)-B(\bvec)}_{\HS}
 \le C|\avec-\bvec|
 \le C' d_H(K_{\avec},K_{\bvec}).
\end{equation}
The same bounds hold for $d_{\mathrm{spec},\infty}$. They also hold with the coefficient norm $\norm{c(\avec)-c(\bvec)}_{\ell_s^2}$ on the right for every $s\ge0$.
\end{theorem}
\begin{proof}
Lemma~\ref{lem:mirsky} and \eqref{eq:singular} give the first inequality. For the second, integrate the uniform Hilbert--Schmidt bound on $DB$ from Lemma~\ref{lem:analytic} along a segment in the parameter ball. The Hausdorff lower bound in Proposition~\ref{prop:geometry} gives the third. The supremum distance is at most the $\ell^2$ distance, and $|\avec-\bvec|\le\norm{c(\avec)-c(\bvec)}_{\ell_s^2}$ for $s\ge0$, which proves the remaining assertions.
\end{proof}

\begin{corollary}[Equivalence of geometric and spectral distances]\label{cor:equivalence}
On a cone from Theorem~\ref{thm:main}, including the vertex, the following distances are mutually comparable for each fixed $s\ge0$:
\[
 |\avec-\bvec|,\quad
 \norm{c(\avec)-c(\bvec)}_{\ell_s^2},\quad
 d_H(K_{\avec},K_{\bvec}),\quad
 |\Lambda_m(\avec)-\Lambda_m(\bvec)|,\quad
 d_{\mathrm{spec},2},\quad d_{\mathrm{spec},\infty}.
\]
The comparison constants may depend on $m,s,F$ and the cone. The same conclusion holds on the sector in Theorem~\ref{thm:m2}.
\end{corollary}
\begin{proof}
Theorem~\ref{thm:main}, Proposition~\ref{prop:geometry}, and Theorem~\ref{thm:forward} give the comparisons, together with
$|\Delta\Lambda_m|\le d_{\mathrm{spec},2}$ and
$|\Delta\Lambda_m|\le\sqrt m\,d_{\mathrm{spec},\infty}$.
\end{proof}

The boundary comparison in Subsection~\ref{sec:geometry} adds $d_H(\partial D_{\avec},\partial D_{\bvec})$ to this list of equivalent distances.

\begin{center}
\small
\begin{tabular}{p{0.25\textwidth}p{0.37\textwidth}p{0.27\textwidth}}
\toprule
Type of bound & Estimate & Parameter restriction\\
\midrule
Coefficients from geometry & $\norm{\Delta c}_{\ell_s^2}\le C_s d_H$ & Small parameter ball\\[4pt]
Forward spectral stability & $d_{\mathrm{spec},2}\le C|\Delta\avec|\le C'd_H$ & Small parameter ball\\[4pt]
Inverse spectral stability & $|\Delta\avec|+d_H\le C|\Delta\Lambda_m|$ & Admissible cone; sector for $m=2$\\
\bottomrule
\end{tabular}
\end{center}
The geometric estimates and forward spectral stability hold throughout a small parameter ball. The inverse recovery result uses the nondegeneracy and angular restriction of Theorem~\ref{thm:main}.

\begin{remark}[Noisy spectral data]
If measured data $y\in\R^m$ satisfy $|y-\Lambda_m(\avec)|\le\epsilon_{\rm data}$ and a fitted parameter $\widehat{\avec}$ in the same cone satisfies $|y-\Lambda_m(\widehat{\avec})|\le\epsilon_{\rm fit}$, then
\[
 |\widehat{\avec}-\avec|\le c^{-1}(\epsilon_{\rm data}+\epsilon_{\rm fit}).
\]
Both the true and fitted parameters must lie in the same prescribed cone. The estimate then bounds recovery error by the sum of the measurement and fitting errors.
\end{remark}

\begin{remark}[Even spectral moments]\label{rem:evenmoments}
The leading eigenvalues vary to first order near the disk. Their even moments vanish to higher order: for $I_k=\sum_j(2\lambda_j)^{2k}$, the bound $B(tv)=O_{\HS}(t)$ gives $I_k(tv)=O(t^{2k})$. A fixed finite list of these moments thus changes by $O(t^2)$ between $tv$ and $2tv$, while the parameter distance is $t|v|$ and every weighted coefficient distance is at least $t|v|$. Hence no such moment list, with Euclidean distance, has a uniform Lipschitz inverse down to the disk. The vector of leading eigenvalues has this property by \eqref{eq:spectralbilip}.
\end{remark}

\section{Two and three Laurent coefficients}\label{sec:examples}
The leading spectral maps can be inverted explicitly for two and three coefficients. The formulas give a full positive sector in dimension two and a cone about an explicit direction in dimension three.

\subsection{Two coefficients}\label{sec:m2}
Take the two-coefficient family
\[
 \Psi_{a,b}(w)=w+a/w+b/w^2+(a^2+b^2)F(w),\qquad F(w)=O(w^{-3}).
\]
The leading matrix, its characteristic coefficients and their Jacobian are
\[
 L_2=\begin{pmatrix}a&\sqrt2b\\\sqrt2b&0\end{pmatrix},\quad
 e_1=a,\quad e_2=-2b^2,\quad
 \det D(e_1,e_2)=-4b.
\]
For $a,b>0$, let $d=\sqrt{a^2+8b^2}$. The signed eigenvalues and singular values are
\[
 \mu_+=(a+d)/2,\quad \mu_-=(a-d)/2,
 \qquad s_1^0=(d+a)/2,\quad s_2^0=(d-a)/2.
\]
The eigenvalues are nonzero, with absolute-value gap $s_1^0-s_2^0=a>0$. The Jacobian is also nonzero, so every positive direction is admissible.

\begin{theorem}[Uniform sector theorem for $m=2$]\label{thm:m2}
Fix $0<\eta<1$ and a tail $F$. For some $\delta>0$, on
\[
 \mathcal Q_{\eta,\delta}
 =\{(a,b):0<a,b<\delta,\ \eta\le b/a\le\eta^{-1}\}
\]
the map $(a,b)\mapsto(\lambda_1,\lambda_2)$ is globally bi-Lipschitz, with constants uniform down to the disk. The geometric and weighted coefficient estimates of Theorem~\ref{thm:main} hold on the same sector.
\end{theorem}
\begin{proof}
The directions in this sector form a compact subset of the positive unit quadrant. Thus \eqref{eq:gaps} and Lemma~\ref{lem:perturb} hold uniformly. To compare two arbitrary points in the sector, use the inverse of the leading spectral map:
\begin{equation}\label{eq:m2inverse}
 a=x-y,\qquad b=\sqrt{xy/2},\qquad (x,y)=(s_1^0,s_2^0).
\end{equation}
The sector condition implies
\[
 1\le x/y\le M_\eta,
 \qquad M_\eta=\frac{\sqrt{1+8\eta^2}+1}{\sqrt{1+8\eta^2}-1}.
\]
The wedge $\{y>0,\ y\le x\le M_\eta y\}$ is convex, and the derivatives of the second component in \eqref{eq:m2inverse} are $(2\sqrt2)^{-1}\sqrt{y/x}$ and $(2\sqrt2)^{-1}\sqrt{x/y}$. Both components have uniformly bounded derivatives on this wedge. Integration along a segment gives a constant $K_\eta$ such that
\[
 |p-q|\le K_\eta|S_0(p)-S_0(q)|\qquad(p,q\in\mathcal Q_{\eta,\delta}).
\]
Let $E=S-S_0$. The derivative estimate \eqref{eq:Serror}, integrated on the convex parameter region, gives $|E(p)-E(q)|\le C\delta|p-q|$. Hence
\[
 |S(p)-S(q)|\ge(K_\eta^{-1}-C\delta)|p-q|.
\]
Choose $\delta$ so that $C\delta\le(2K_\eta)^{-1}$. This gives the lower bound; the uniform bounds on $DS_0$ and \eqref{eq:Serror} give the upper bound. Equation \eqref{eq:singular} transfers both bounds to the NP eigenvalues, with a factor of two. Proposition~\ref{prop:geometry} gives the geometric and coefficient estimates.
\end{proof}

\subsubsection{Leading reconstruction and an infinite-tail example}
The explicit inverse suggests the reconstruction formulas
\begin{equation}\label{eq:m2approx}
 a_{\rm lead}=2(\lambda_1-\lambda_2),\qquad
 b_{\rm lead}=\sqrt{2\lambda_1\lambda_2}.
\end{equation}
They satisfy $(a_{\rm lead},b_{\rm lead})=(a,b)+O(a^2+b^2)$ uniformly on the sector. Indeed, $S-S_0=O(|p|^2)$, and the inverse \eqref{eq:m2inverse} has a bounded derivative on a slightly enlarged positive wedge containing both spectral vectors. Thus the leading reconstruction is accurate to second order.

For an example with infinitely many coefficients, take
\[
 F_2(w)=e^{1/w}-1-w^{-1}-\tfrac12w^{-2}
 =\sum_{n\ge3}\frac{w^{-n}}{n!}.
\]
On the sector $1/2\le b/a\le2$, the sufficient condition \eqref{eq:univalent} becomes
\[
 a+2b+(e-2)(a^2+b^2)<1.
\]
The inverse theorem may require a smaller parameter bound. The weighted coefficient estimate has $F_s^2=\sum_{n\ge3}n^{2s}/(n!)^2$. Figure~\ref{fig:m2} shows domains in this family.

\begin{figure}[htbp]
\revisionfloat
\centering
\includegraphics[width=\textwidth]{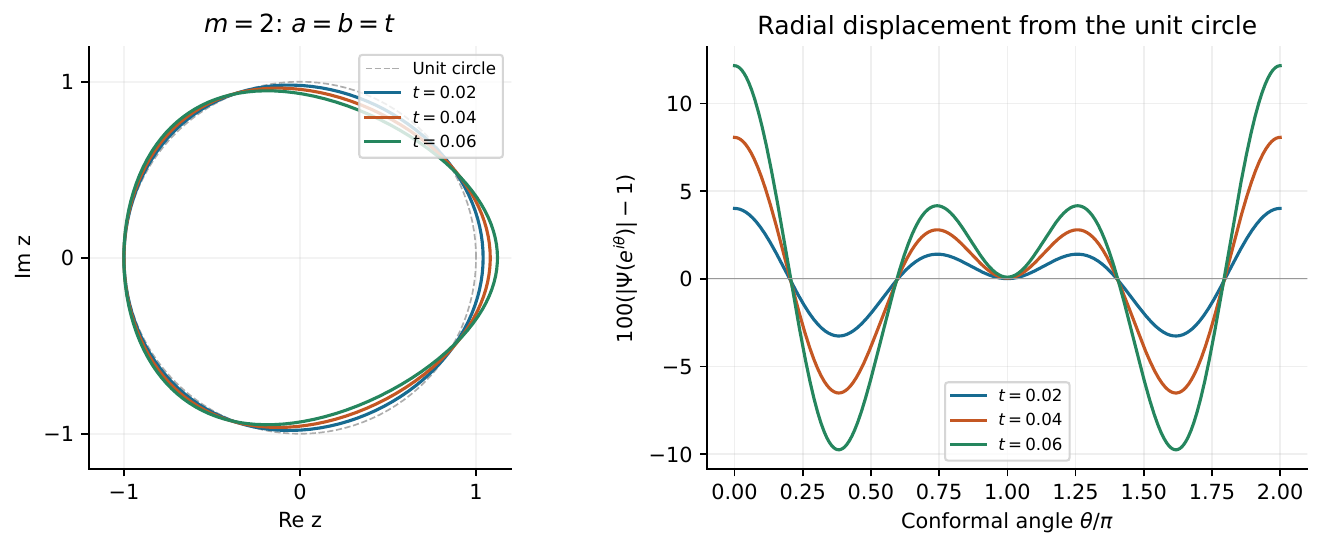}
\caption{Two-coefficient domains with the exponential tail $F_2$ and $a=b=t$, for $t=0.02,0.04,0.06$. Left: boundaries and the unit circle. Right: $100(|\Psi(e^{i\theta})|-1)$ versus conformal angle. All curves shown are strictly convex by Subsection~\ref{sec:exampleconvexity}; the inverse theorem holds on a sufficiently small interval of the ray.}\label{fig:m2}
\end{figure}

\FloatBarrier
\subsection{Three coefficients}\label{sec:m3}
Now take three coefficients:
\[
 \Psi_{a,b,c}(w)=w+a/w+b/w^2+c/w^3+(a^2+b^2+c^2)F(w),
 \qquad F(w)=O(w^{-4}).
\]
Its leading matrix is
\begin{equation}\label{eq:L3}
 L_3(a,b,c)=\begin{pmatrix}
 a&\sqrt2b&\sqrt3c\\\sqrt2b&2c&0\\\sqrt3c&0&0
 \end{pmatrix}.
\end{equation}
Let $T=\Tr L_3$, $Q=\Tr L_3^2$, and $\Delta_3=\det L_3$. Then
\begin{equation}\label{eq:TQD}
 T=a+2c,\qquad Q=a^2+4b^2+10c^2,\qquad \Delta_3=-6c^3.
\end{equation}
The Jacobian is
\begin{equation}\label{eq:jac3}
 \det\frac{\partial(T,Q,\Delta_3)}{\partial(a,b,c)}
 =\det\begin{pmatrix}1&0&2\\2a&8b&20c\\0&0&-18c^2\end{pmatrix}
 =-144bc^2.
\end{equation}
The change of variables $(e_1,e_2,e_3)=(T,(T^2-Q)/2,\Delta_3)$ has Jacobian determinant $-1/2$. Hence $\det D\mathcal E_3=72bc^2$, in agreement with the general induction.

At $(a,b,c)=(1,1,1)$ the characteristic polynomial is
\[
 p_3(z)=z^3-3z^2-3z+6.
\]
Evaluate it at six rational points:
\[
 \begin{array}{c|rrrrrr}
 z&-8/5&-3/2&1&13/10&3&7/2\\\hline
 p_3(z)&-122/125&3/8&1&-773/1000&-3&13/8
 \end{array}
\]
The sign changes give a root in each of $(-8/5,-3/2)$, $(1,13/10)$, and $(3,7/2)$. These are all three roots and are simple, nonzero and distinct in absolute value. By \eqref{eq:jac3}, the direction $v=(1,1,1)/\sqrt3$ is therefore admissible.

\begin{corollary}[Three-eigenvalue recovery]\label{cor:m3}
Fix a real exponentially decaying tail beginning at $w^{-4}$. For some $\varepsilon,\delta>0$, on
\[
 \mathcal C\big((1,1,1)/\sqrt3,\varepsilon,\delta\big)
\]
the first three positive reduced NP eigenvalues determine $(a,b,c)$ uniquely, and
\[
 |(a,b,c)-(\widetilde a,\widetilde b,\widetilde c)|
 \le C\left(\sum_{j=1}^3|\lambda_j-\widetilde\lambda_j|^2\right)^{1/2}.
\]
The geometric and weighted coefficient estimates of Theorem~\ref{thm:main} also hold, uniformly down to the disk.
\end{corollary}
\begin{proof}
Apply Theorem~\ref{thm:main}, whose hypotheses we have just checked, and narrow the cone to keep all coordinates positive.
\end{proof}

\subsubsection{Signed ordering and leading reconstruction}
At $(1,1,1)$, the signed eigenvalues, ordered by decreasing absolute value, are $s_1^0,-s_2^0,s_3^0$. This ordering persists in a sufficiently narrow cone. Writing $x_j=s_j^0$, invert \eqref{eq:TQD} explicitly:
\begin{align}
 c&=(x_1x_2x_3/6)^{1/3},\label{eq:m3inv}\\
 a&=x_1-x_2+x_3-2c,\nonumber\\
 b&=\tfrac12\sqrt{x_1^2+x_2^2+x_3^2-a^2-10c^2}.\nonumber
\end{align}
Positivity in the prescribed cone selects the positive root for $b$. This recovers $(a,b,c)$ exactly from the singular values of $L_3$. Substituting $x_j=2\lambda_j$ for the actual domain gives an error $O(a^2+b^2+c^2)$ uniformly on a smaller cone. The inverse in \eqref{eq:m3inv} is smooth and homogeneous of degree one near the leading spectral ray, with uniformly bounded differential. After division by $|(a,b,c)|$, \eqref{eq:Serror} bounds the difference between the two spectral vectors by $O(|(a,b,c)|)$. Shrink the cone and radius so that both vectors and their joining segment stay in this neighbourhood. The radicand remains positive, and integrating the differential gives the error bound.

The corresponding exponential tail is
\[
 F_3(w)=e^{1/w}-1-w^{-1}-\tfrac12w^{-2}-\tfrac16w^{-3}
 =\sum_{n\ge4}\frac{w^{-n}}{n!}.
\]
For positive coefficients, \eqref{eq:univalent} becomes
\[
 a+2b+3c+(e-\tfrac52)(a^2+b^2+c^2)<1.
\]
The inverse theorem may require a smaller radius. The weighted coefficient estimate uses $F_s^2=\sum_{n\ge4}n^{2s}/(n!)^2$. Figure~\ref{fig:m3} shows domains along the admissible ray.

\begin{figure}[htbp]
\revisionfloat
\centering
\includegraphics[width=\textwidth]{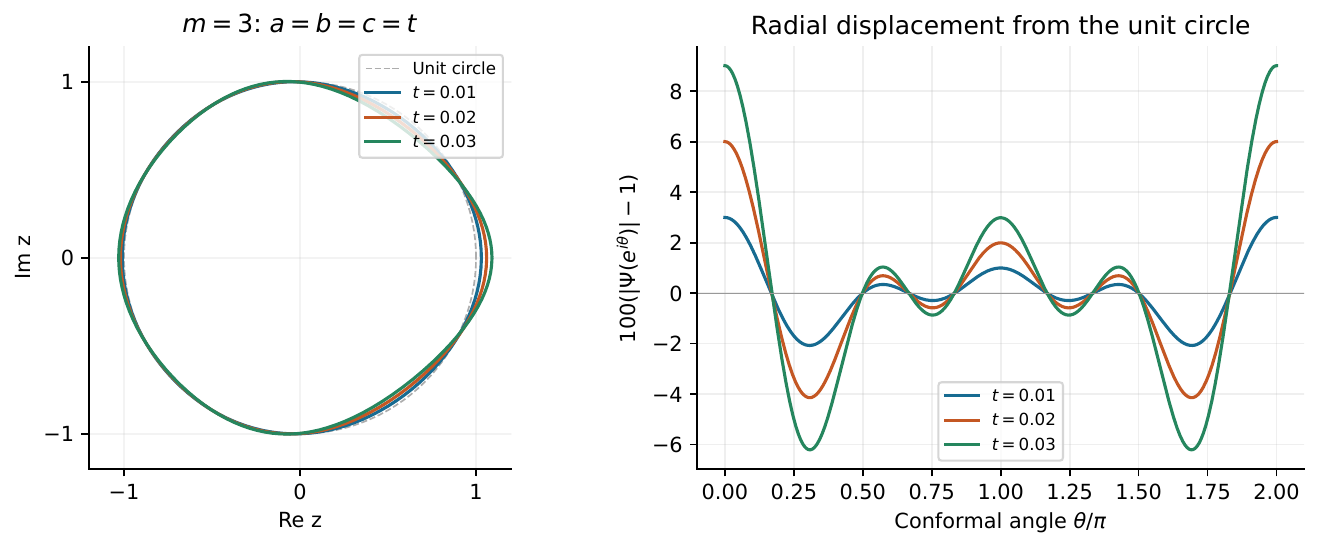}
\caption{Three-coefficient domains with the exponential tail $F_3$ and $a=b=c=t$, for $t=0.01,0.02,0.03$. The panels follow Figure~\ref{fig:m2}. All displayed curves are strictly convex by Subsection~\ref{sec:exampleconvexity}. Corollary~\ref{cor:m3} holds on a sufficiently small interval of the ray.}\label{fig:m3}
\end{figure}

\FloatBarrier
\subsection{Convexity of the displayed examples}\label{sec:exampleconvexity}
Here is a direct convexity check for the curves in Figures~\ref{fig:m2}--\ref{fig:m3}. Let $q_1=\sum n|c_n|$ and $q_2=\sum n^2|c_n|$. The condition $q_1<1$ ensures exterior injectivity, and differentiation of the boundary parametrisation gives
\[
 \operatorname{Im}(\overline{z'}z'')\ge1-q_1-q_2-q_1q_2
\]
for $z(\theta)=\Psi(e^{i\theta})$. At the largest displayed $t$, the exponential sums give lower bounds of $0.43$ for $m=2$ and $0.32$ for $m=3$. The bounds persist for smaller $t$, proving strict convexity of every displayed curve.

\FloatBarrier
\section{Concluding remarks}\label{sec:conclusion}

The NP spectrum can fail to distinguish domains even when their boundaries are analytic, strictly convex and arbitrarily close to the circle, and even after allowing extended Möbius transformations. Within prescribed finite-dimensional conformal families, however, finitely many eigenvalues give Lipschitz stable recovery on suitable truncated cones, with constants uniform down to the disk.

To recover finer boundary features as $m$ grows, we need quantitative bounds on the cone aperture, admissible radius and inverse Lipschitz constant. Spectral cluster observables accommodate repeated eigenvalues near symmetric domains; recovery of several parameters still requires a differential of full rank. The critical curve in Appendix~\ref{app:moments} raises another question: which additional spectral observables separate the two branches of the moment map?

For recovery from resonance measurements, we must also extract the eigenvalues stably and control their excitation and observation weights. This is where the spectral estimates meet the polarization data of Appendix~\ref{app:direct}.

\appendix
\section{The two-moment map}\label{app:moments}
\subsection{Uniqueness on a positive sector}
For $m=2$, the full spectral moments are
\[
 I_1=4\sum_j\lambda_j^2=\Tr B^2,\qquad
 I_2=16\sum_j\lambda_j^4=\Tr B^4.
\]
Lemma~\ref{lem:analytic} shows that these moments are real analytic in the parameters. Writing $\rho=|a|+|b|$, we have
\begin{equation}\label{eq:oldL}
 B=L+Q+O_{\HS}(\rho^3),\qquad
 L=\begin{pmatrix}a&\sqrt2b\\\sqrt2b&0\end{pmatrix}\oplus0,
\end{equation}
where $Q$ is homogeneous of degree two. The logarithmic series, together with the prescribed tail beginning at $w^{-3}$, shows that the nonzero entries of $L$ have $i+j\le3$, whereas those of $Q$ have $i+j\ge4$. Thus, $\Tr(LQ)=0$, and
\begin{align}
 I_1&=a^2+4b^2+O(\rho^4),\label{eq:oldI1}\\
 I_2&=a^4+8a^2b^2+8b^4+O(\rho^5),\label{eq:oldI2}\\
 J_F:=\det\frac{\partial(I_1,I_2)}{\partial(a,b)}
 &=-64ab^3+O(\rho^5).\label{eq:oldJ}
\end{align}
Expanding $\Tr B^2$ and using the disjoint supports gives the formula for $I_1$. For $I_2$, the leading term is $\Tr L^4$; every other term has one factor of at least quadratic order and three factors of at least linear order, so its contribution is $O(\rho^5)$. Differentiating the analytic expansions gives
\[
 I_{1,a}=2a+O(\rho^3),\quad I_{1,b}=8b+O(\rho^3),
\]
\[
 I_{2,a}=4a^3+16ab^2+O(\rho^4),\quad
 I_{2,b}=16a^2b+32b^3+O(\rho^4).
\]
Substituting these derivatives in the determinant gives \eqref{eq:oldJ}.

\begin{lemma}[Connected monotone level sets]\label{lem:global}
Let $u,v$ be $C^1$ on a neighbourhood of the sector $\mathcal Q_{\eta,\delta}$ in Theorem~\ref{thm:m2}. If $u_a,u_b>0$ and $\det\partial(u,v)/\partial(a,b)<0$ throughout the sector, then $(u,v)$ is injective on the sector.
\end{lemma}
\begin{proof}
Suppose two points have the same value of $u$. If their $a$ coordinates agree, strict monotonicity in $b$ makes the points equal. Otherwise, order them so that $a_1<a_2$. Strict monotonicity along segments with increasing coordinates forces $b_1>b_2$. The rectangle $[a_1,a_2]\times[b_2,b_1]$ lies in the sector, since its extreme ratios are $b_2/a_2$ and $b_1/a_1$. Writing their common value as $u=c$, we have, for every intermediate $a$,
\[
 u(a,b_2)\le c\le u(a,b_1).
\]
Strict monotonicity in $b$ therefore gives a unique graph $b_c(a)$ through both points with $u(a,b_c(a))=c$. By the implicit function theorem, $b_c'=-u_a/u_b<0$, and
\begin{equation}\label{eq:levelderivative}
 \frac{\dd}{\dd a}v(a,b_c(a))
 =-\frac{\det\partial(u,v)/\partial(a,b)}{u_b}>0.
\end{equation}
Hence, $v$ is strictly increasing along this level set, and distinct endpoints have different values of $v$.
\end{proof}

\begin{theorem}[Global uniqueness from two moments]\label{thm:twomoments}
For each fixed tail and each $0<\eta<1$, the map $(a,b)\mapsto(I_1,I_2)$ is injective on $\mathcal Q_{\eta,\delta}$ for sufficiently small $\delta>0$.
\end{theorem}
\begin{proof}
Throughout the sector, $a,b\ge c_\eta\rho$ and $ab^3\ge\kappa_\eta\rho^4$, where
\[
 c_\eta=\frac\eta{1+\eta},\qquad
 \kappa_\eta=\min_{\eta\le t\le\eta^{-1}}\frac{t^3}{(1+t)^4}>0.
\]
The differentiated expansions therefore give $I_{1,a},I_{1,b}>0$ and $J_F<0$, uniformly for all sufficiently small $\rho>0$. Lemma~\ref{lem:global}, applied with $u=I_1,v=I_2$, proves injectivity.
\end{proof}
Theorem~\ref{thm:twomoments} gives uniqueness from two full spectral moments, complementing the two-eigenvalue result of Theorem~\ref{thm:m2}. The moment pair $(I_1,I_2)$ has no uniformly Lipschitz inverse down to the disk, as Remark~\ref{rem:evenmoments} shows. For the family with $F=0$, a further obstruction appears outside the comparable-coefficient sector.

\subsection{The critical curve for the family without a tail}\label{sec:fold}
For $F=0$, the next term in the Jacobian expansion reveals a fold in the two-moment map. The map is injective on each side, giving larger uniqueness regions, and fails to be injective on the entire small positive quadrant.

\subsubsection{Expansion of the moment map}
The upper $2\times2$ block of the quadratic term $Q$ in \eqref{eq:oldL} has only one nonzero entry, $Q_{22}=a^2$. With $u=w^{-1}$ and $v_0=z^{-1}$, the quadratic term of the logarithm is
\[
 \tfrac12\bigl(auv_0+buv_0^2+bu^2v_0\bigr)^2,
\]
whose coefficient of $u^2v_0^2$ is $a^2/2$. Since $(L^3)_{22}=2ab^2$,
\[
 4\Tr(L^3Q)=8a^3b^2.
\]
We obtain
\begin{align}
 I_1&=a^2+4b^2+O(\rho^4),\label{eq:foldI1}\\
 I_2&=a^4+8a^2b^2+8b^4+8a^3b^2+O(\rho^6),\label{eq:foldI2}\\
 J_0&=-64ab^3+32a^4b-192a^2b^3+O(\rho^6).\label{eq:foldJ}
\end{align}
Differentiating $8a^3b^2$ gives the terms of degree five in the determinant. The remainder of degree four in $I_1$ can contribute to $J_0$ only at degree six or higher.

\subsubsection{The critical curve}
For $F=0$ and complex $a,b$ near zero, define
\[
 \mathcal I_k(a,b)=\Tr\bigl((B(a,b)^*B(a,b))^k\bigr),\qquad k=1,2.
\]
The proof of Lemma~\ref{lem:analytic} shows that $B$ is holomorphic as a Hilbert--Schmidt-valued map. Hence, $\mathcal I_k$ is real analytic and coincides with $I_k$ for real $a,b$. Rotation through $\theta$ gives
\[
 B(e^{2i\theta}a,e^{3i\theta}b)=D_\theta B(a,b)D_\theta,
 \qquad D_\theta=\diag(e^{in\theta})_{n\ge1}.
\]
Unitary invariance gives $I_k(a,-b)=I_k(a,b)$ on taking $\theta=\pi$. At $a=0$, rotation through $2\pi/3$ fixes $b$ and rotates the complex $a$-plane through $4\pi/3$. The only invariant real differential on that plane is zero, so $\partial_aI_k(0,b)=0$.

Thus, $J_0$ is odd in $b$ and vanishes at $a=0$. There is an analytic function $q(a,y)$ near $(0,0)$ such that
\begin{equation}\label{eq:qfactor}
 J_0(a,b)=ab\,q(a,b^2).
\end{equation}
Equation~\eqref{eq:foldJ} gives
\begin{equation}\label{eq:qexp}
 q(a,y)=32a^3-64y-192ay+O(a^4+a^2|y|+y^2).
\end{equation}
Since $q_y(0,0)=-64$, the analytic implicit function theorem gives a unique function $h$ near zero satisfying
\begin{equation}\label{eq:h}
 q(a,h(a))=0,\qquad h(a)=\tfrac12a^3+O(a^4).
\end{equation}
For sufficiently small positive $a$, $h(a)>0$ and $h'(a)=\tfrac32a^2+O(a^3)>0$. After shrinking the neighbourhood so that $q_y<0$, the analytic factorisation
\[
 q(a,y)=u(a,y)(h(a)-y),\qquad u(a,y)>0,
\]
gives
\begin{equation}\label{eq:sign}
 \operatorname{sgn}J_0(a,b)=\operatorname{sgn}(h(a)-b^2)
 \quad(a,b>0\text{ small}).
\end{equation}
Hence, $b^2=h(a)$ is the only critical curve in the small positive quadrant.

\begin{theorem}[Global uniqueness on either side of the fold]\label{thm:branches}
There is $\delta_0>0$ such that the two-moment map is injective on each of the regions
\[
 U_- =\{0<a,b<\delta_0:\ b^2>h(a)\},\qquad
 U_+ =\{0<a,b<\delta_0:\ b^2<h(a)\}.
\]
Thus, the NP spectrum determines the normalised domain with $F=0$ within either region. With $\delta_0$ smaller if necessary, injectivity holds separately on the two explicit regions
\[
 0<a,b<\delta_0,\quad b^2\ge a^3,
 \qquad\text{and}\qquad
 0<a,b<\delta_0,\quad b^2\le\tfrac14a^3.
\]
\end{theorem}
\begin{proof}
For real $a,b>0$, every matrix entry in the logarithmic expansion has nonnegative coefficients in $a,b$. The identities $B_{11}=a$ and $B_{12}=B_{21}=\sqrt2b$ therefore give
\[
 I_{1,a}\ge2a>0,\qquad I_{1,b}\ge8b>0.
\]
Distinct points on the same $I_1$ level have oppositely ordered coordinates. The argument of Lemma~\ref{lem:global} joins them by a graph within their coordinate rectangle, with $b_c(a)$ strictly decreasing. Since $h$ increases, $b_c(a)^2-h(a)$ is strictly decreasing too. If both endpoints lie on one side of the critical curve, the entire joining graph stays on that side. Equations~\eqref{eq:levelderivative} and \eqref{eq:sign} then imply strict monotonicity of $I_2$ along the graph, proving injectivity. Finally, $h(a)=a^3/2+O(a^4)$ puts the two explicit regions on the indicated sides for sufficiently small~$a$.
\end{proof}

\begin{proposition}[The two moments do not distinguish the whole quadrant]\label{prop:fold}
For every sufficiently small $a_0>0$, every neighbourhood of $(a_0,b_0)$, where $b_0=\sqrt{h(a_0)}>0$, contains distinct parameter pairs on opposite sides of $b^2=h(a)$ with equal $I_1$ and equal $I_2$.
\end{proposition}
\begin{proof}
The level set of $I_1$ through $(a_0,b_0)$ is locally the graph $b=b_c(a)$, where $b_c'=-I_{1,a}/I_{1,b}<0$. The factorisation $J_0(a,b)=ab\,u(a,b^2)(h(a)-b^2)$ gives
\[
 \left.\frac{\dd}{\dd a}J_0(a,b_c(a))\right|_{a=a_0}
 =a_0b_0u(a_0,b_0^2)\bigl(h'(a_0)-2b_0b_c'(a_0)\bigr)>0.
\]
By Equation~\eqref{eq:levelderivative}, the second derivative of $I_2(a,b_c(a))$ at $a_0$ is strictly negative. Thus each nearby value below this nondegenerate local maximum is attained once on each side of the critical curve.
\end{proof}
These pairs have equal values of two moments; higher moments may distinguish them. Hence the two-moment map is not injective on $U_-\cup U_+$.

Uniform spectral stability throughout these regions requires further bounds, since the angular gaps and inverse-Jacobian constants used in the spectral argument can deteriorate near degenerate directions.

\section{Uniqueness from polarization data}\label{app:direct}
\subsection{Spectral weights and polarization data}
Complete weighted NP data recover all area moments and hence the domain. At one contrast, all uncontracted GPTs also give uniqueness for analytic Jordan domains. This second proof uses all polynomial observations; \cite[Theorem~3.1]{APSTident} requires only harmonic combinations. The spectral weights record the excitation and observation couplings of energy-normalised eigenfunctions.

Let $D$ be a bounded smooth planar domain whose reduced NP operator has the self-adjoint energy realisation $A_D$. The smooth Jordan domains of the main text satisfy this assumption. For a harmonic polynomial $H$ and a polynomial $G$, define
\[
 g_H=\partial_\nu H,\qquad
 \ell_G(\varphi)=\int_{\partial D}G\varphi\dd s.
\]
The input has zero mean. Equivalence of the energy norm with the mean-zero $H^{-1/2}$ norm makes $\ell_G$ continuous on the energy space. We extend to complex polynomials by complexification, using bilinear boundary pairings with no conjugation of $G$.

Let $E_\mu$ be the orthogonal spectral projection of $A_D$ associated with an eigenvalue $\mu$, including $E_0$ onto the kernel. The geometric spectral weights are
\begin{equation}\label{eq:weights}
 W_{G,H}(\mu;D)=\ell_G(E_\mu g_H).
\end{equation}
If $(\varphi_{\mu,j})_j$ is an energy-orthonormal eigenbasis, with the energy inner product linear in its first argument, then
\[
 W_{G,H}(\mu;D)=\sum_j
 (g_H,\varphi_{\mu,j})_*\,\ell_G(\varphi_{\mu,j}).
\]
These weights are independent of the chosen basis and are determined by the eigenvalues, energy-normalised eigenfunctions and their geometric couplings.

\subsection{A direct area-moment proof from weighted spectral data}
\begin{theorem}[Uniqueness from weighted spectral data]\label{thm:weighted}
The reduced NP eigenvalues and the weights \eqref{eq:weights}, including the weights of the zero eigenspace, for all harmonic polynomials $H,G$ uniquely determine a bounded smooth planar domain in fixed Euclidean coordinates. 
\end{theorem}
\begin{proof}
The spectral theorem and the continuity of $\ell_G$ give
\begin{equation}\label{eq:weightedresolvent}
 M_{H,G}(\lambda;D)=\ell_G((\lambda I-A_D)^{-1}g_H)
 =\sum_\mu\frac{W_{G,H}(\mu;D)}{\lambda-\mu}.
\end{equation}
Absolute summability of the weights follows by representing $\ell_G$ by its Riesz vector and applying Cauchy--Schwarz to the orthogonal spectral components of that vector and $g_H$. Moreover, in operator norm,
\[
 \lambda(\lambda I-A_D)^{-1}=(I-A_D/\lambda)^{-1}\longrightarrow I
 \quad(\lambda\to+\infty).
\]
Green's formula then gives
\begin{equation}\label{eq:momentlimit}
 \lim_{\lambda\to+\infty}\lambda M_{H,G}(\lambda;D)
 =\int_{\partial D}G\partial_\nu H\dd s
 =\int_D\nabla G\cdot\nabla H\dd A,
\end{equation}
since $\Delta H=0$.

For integers $p,q\ge0$ choose
\[
 H_p(z)=\frac{z^{p+1}}{p+1},\qquad
 G_q(z)=\frac{\bar z^{q+1}}{q+1}.
\]
These are complex harmonic polynomials satisfying
$\nabla H_p\cdot\nabla G_q=2z^p\bar z^q$, so
\begin{equation}\label{eq:allmixedmoments}
 \int_D z^p\bar z^q\dd A
 =\frac12\lim_{\lambda\to+\infty}\lambda M_{H_p,G_q}(\lambda;D)
 =\frac12\sum_\mu W_{G_q,H_p}(\mu;D).
\end{equation}
By bilinear complexification, real harmonic data suffice to recover these mixed moments.

If two domains have the same weighted spectral data, their area measures agree on all polynomials in $z,\bar z$, hence on all real polynomials in $x,y$. On a compact rectangle containing both closures, these polynomials are uniformly dense in the continuous functions by Stone--Weierstrass. The area measures therefore coincide. Each support is the closure of its domain, and smoothness gives $D=\operatorname{int}\overline D$, so the domains are equal.
\end{proof}

\begin{remark}
By Formula~\eqref{eq:allmixedmoments}, the sums of the complete weights suffice to determine the domain. The eigenvalues affect the full resolvent response; its leading coefficient at infinity depends only on these sums. The kernel contribution is essential: the reduced NP operator of the disk is zero, but its weights recover nonzero area moments. Here the complete spectral decomposition allows $\lambda$ to vary. The fixed-contrast argument below uses different data.
\end{remark}

\subsection{A direct proof for GPTs at one fixed contrast}
At a fixed contrast, the moments of all polynomial observations determine the boundary charge measures. This yields uniqueness in the analytic Jordan class once we establish density of the harmonic inputs.

\begin{lemma}[Density of harmonic-polynomial Neumann traces]\label{lem:neumanndense}
For a bounded simply connected planar domain with analytic Jordan boundary,
\[
 \operatorname{span}\{\partial_\nu H:H\text{ a harmonic polynomial}\}
\]
is dense in $L^2_0(\partial D)$.
\end{lemma}
\begin{proof}
It suffices to approximate a smooth real mean-zero function $f$, since these functions are dense in $L^2_0$. Solve the Neumann problem for a harmonic $u\in C^\infty(\overline D)$ with $\partial_\nu u=f$. By simple connectedness, there is a holomorphic function $h$, smooth up to the boundary, with real part $u$. Let $\psi:\mathbb D\to D$ be an interior conformal map and $\phi=\psi^{-1}$. Analyticity of the Jordan boundary extends these conformal maps to collars of the closures. For $0<r<1$ sufficiently close to one,
\[
 h_r(z)=h(\psi(r\phi(z)))
\]
is holomorphic in a neighbourhood of $\overline D$ and converges to $h$ in $C^1(\overline D)$ as $r\uparrow1$. Since the complement of $\overline D$ is connected, Runge approximation on the closure of a slightly larger Jordan domain in this neighbourhood, followed by Cauchy estimates, gives polynomial approximations to $h_r$ and its first derivative, uniform on $\overline D$. Taking real parts approximates $u$ in $C^1(\overline D)$, so the normal derivatives converge uniformly, hence in $L^2$, to $f$. Linearity handles complex data.
\end{proof}

\begin{theorem}[Fixed-contrast uniqueness from uncontracted GPTs]\label{thm:fixedGPT}
Let $D_1,D_2$ be bounded analytic Jordan domains, and let $\lambda^{(1)},\lambda^{(2)}\in\R\setminus[-1/2,1/2]$. If
\[
 \int_{\partial D_1}G(\lambda^{(1)}I-K_{D_1}^*)^{-1}[\partial_\nu H]\dd s
 =\int_{\partial D_2}G(\lambda^{(2)}I-K_{D_2}^*)^{-1}[\partial_\nu H]\dd s
\]
for every harmonic polynomial $H$ and every polynomial $G$, then $D_1=D_2$ and $\lambda^{(1)}=\lambda^{(2)}$.
\end{theorem}
\begin{proof}
Linear combinations of input and observation monomials show that equality of all uncontracted GPTs at the respective contrasts gives the stated hypothesis. For each $H$ define
\[
 \varphi_{j,H}=(\lambda^{(j)}I-K_{D_j}^*)^{-1}[\partial_\nu H],
 \qquad \mu_{j,H}=\varphi_{j,H}\,\dd s_{\partial D_j}.
\]
The resolvents are bounded isomorphisms on $L^2_0(\partial D_j)$, so the response densities lie in $L^2$ and have zero integral. The data give identical polynomial moments for the compactly supported finite signed or complex measures $\mu_{1,H},\mu_{2,H}$. Their constant moments vanish as well. By Stone--Weierstrass,
\[
 \mu_{1,H}=\mu_{2,H}\quad\text{as measures on the ambient plane, for every }H.
\]

By Lemma~\ref{lem:neumanndense} and invertibility of the resolvents, the responses $\varphi_{j,H}$ span a dense subspace of $L^2_0(\partial D_j)$. Suppose $x\in\partial D_1\setminus\partial D_2$. A sufficiently small ball about $x$ avoids $\partial D_2$ and contains an open arc of $\partial D_1$. Equality of the charge measures makes every $\varphi_{1,H}$ vanish almost everywhere on that arc. Every function in their closed span would then vanish there, contradicting the existence of a nonzero smooth mean-zero function supported on the arc. Hence, $\partial D_1\subset\partial D_2$. Interchanging the domains gives equality of their boundaries and therefore of their bounded Jordan interiors.

On the common boundary, equality of the charge measures gives $\varphi_{1,H}=\varphi_{2,H}=\varphi_H$. Subtracting the two resolvent equations gives $(\lambda^{(1)}-\lambda^{(2)})\varphi_H=0$. For $H=x_1$, the normal derivative cannot vanish identically on a bounded smooth domain, so invertibility gives $\varphi_H\ne0$. Thus, $\lambda^{(1)}=\lambda^{(2)}$. The common value $\lambda=\lambda^{(1)}=\lambda^{(2)}$ gives the conductivity $k=(2\lambda+1)/(2\lambda-1)$.
\end{proof}

\section{Additional uniqueness classes and extensions}\label{sec:enlarge}
The class in Theorem~\ref{thm:main} admits several extensions: other finite-dimensional slices with injective spectral differential, symmetry-reduced families, and saturation under spectrum-preserving transformations. Saturation gives uniqueness on a quotient. Throughout this appendix, eigenvalues are ordered with multiplicity, as in Subsection~\ref{sec:spectralmetrics}.

\subsection{Ellipses: an exact global one-parameter class}
\begin{proposition}[Normalised ellipses]
Let $E_t$ have exterior map $\Psi_t(w)=w+t/w$, $0\le t<1$. Its semiaxes are $1+t$ and $1-t$, and its positive reduced NP eigenvalues are
\[
 \lambda_n(E_t)=\tfrac12t^n,\qquad n\ge1.
\]
At $t=0$ we use the zero-padded eigenvalue sequence. Hence
\[
 |t-u|=2|\lambda_1(E_t)-\lambda_1(E_u)|,
 \qquad d_H(\overline E_t,\overline E_u)=|t-u|.
\]
The NP spectrum therefore determines an ellipse up to similarities, with global Lipschitz stability of its capacity-normalised semiaxes. The axis ratio $(1+t)/(1-t)$ is uniformly Lipschitz in $\lambda_1$ only on intervals $t\le t_*<1$.
\end{proposition}
\begin{proof}
The logarithmic kernel is
\[
 -\log\frac{\Psi_t(w)-\Psi_t(\zeta)}{w-\zeta}
 =-\log(1-t/(w\zeta))
 =\sum_{n\ge1}\frac{t^n}{n}w^{-n}\zeta^{-n}.
\]
The normalised Grunsky matrix is therefore $\diag(t,t^2,\ldots)$, and \eqref{eq:singular} gives the eigenvalues. Corresponding boundary points differ by $(t-u)e^{-i\theta}$, so the Hausdorff distance between the convex bodies is at most $|t-u|$. Their horizontal support values differ by $|t-u|$, proving equality. The derivative of the axis ratio is $2/(1-t)^2$.
\end{proof}

\subsection{Single-mode families, including repeated leading eigenvalues}\label{sec:singlemode}
For fixed $k$ and $t\ge0$, the family in \cite{APRYZ}, equation~(2.9), has capacity-normalised map $\Psi_t(w)=w+t w^{-k}$. The integer $k$ specifies a mode; $m$ in the main theorem counts independent coefficients. To allow repeated singular values, we use an observable for the whole leading spectral cluster.

\begin{theorem}[Exact single-mode recovery]\label{thm:singlemode}
Fix $k\ge1$ and consider the family with a prescribed tail
\[
 \Psi_t(w)=w+t w^{-k}+\mathcal R(t;w),\qquad
 \mathcal R(t;w)=\sum_{n>k}r_n(t)w^{-n},
\]
where $\mathcal R$ is jointly holomorphic near $t=0$ and for $|w|>r_*$, with $0<r_*<1$. Its Laurent coefficients are real for real $t$, and it is uniformly $O(t^2)$ on a circle $r_*<|w|=r<1$. There are $\delta,c,C>0$ such that the domains are analytic and strictly convex for $0\le t\le\delta$ and, for all $t,u\in[0,\delta]$,
\begin{equation}\label{eq:singlelip}
 c|t-u|\le
 \left(\sum_{j=1}^k|\lambda_j(D_t)-\lambda_j(D_u)|^2\right)^{1/2}
 \le C|t-u|.
\end{equation}
The single scalar observable
\[
 \mathcal J_k(t)=\left(\sum_{j=1}^k\lambda_j(D_t)^2\right)^{1/2}
\]
is strictly increasing and has a Lipschitz inverse, uniformly down to the disk. For $t,u\in[0,\delta]$, the boundary Hausdorff distance and the coefficient distance in each fixed weighted space $\ell_s^2$, $s\ge0$, are equivalent to $|t-u|$. The coefficient constants may depend on $s$.
\end{theorem}
\begin{proof}
The estimates in Subsections~\ref{sec:geometry} and~\ref{sec:analytic} apply to this family. In the real self-adjoint Grunsky realisation,
\[
 B(t)=tA_k+O_{\HS}(t^2),\qquad
 (A_k)_{ij}=\begin{cases}\sqrt{ij},&i+j=k+1,\\0,&\text{otherwise}.
 \end{cases}
\]
The remainder has derivative $O_{\HS}(t)$. The matrix $A_k$ has rank $k$, with $k$ nonzero singular values that may repeat, and
\[
 \chi_k=\norm{A_k}_{\HS}^2
 =\sum_{j=1}^k j(k+1-j)=\frac{k(k+1)(k+2)}6.
\]
Define $C(t)=B(t)/t$ by analytic extension, with $C(0)=A_k$. Choose fixed disjoint contours around the nonzero eigenvalues of $A_k$, excluding zero. Their combined Riesz projection $P(t)$ is analytic of rank $k$. For small real $t$, it selects the $k$ eigenvalues of $C(t)$ with largest absolute value, counted with multiplicity. Thus
\[
 f(t)=\Tr(C(t)^2P(t))=\chi_k+O(t)
\]
is analytic, real and positive near zero, with $f'(t)=O(1)$. For $t\ge0$,
\[
 \mathcal J_k(t)=\tfrac t2\sqrt{f(t)},\qquad
 \mathcal J_k'(t)=\tfrac12\sqrt{\chi_k}+O(t).
\]
Reducing $\delta$ bounds the derivative above and below by positive constants. The mean-value theorem and reverse triangle inequality give the lower bound in \eqref{eq:singlelip}. The upper bound follows from the Hilbert--Schmidt singular-value inequality and $\norm{B(t)-B(u)}_{\HS}\le C|t-u|$.

Corresponding boundary points give the geometric upper bound. For the lower bound, use the support-function expansion
\[
 h_t(\theta)=1+t\cos((k+1)\theta)+q_t(\theta),\qquad
 \norm{q_t-q_u}_{C^0}\le C\delta|t-u|.
\]
This follows from the calculation in Proposition~\ref{prop:geometry}: the linear normal displacement is the displayed cosine, and the remainder has derivative $O(t)$. Taking the supremum and reducing $\delta$ gives the lower bound for the convex bodies, hence also for their boundaries. The prescribed tail has coefficient derivative $O(t)$ in every fixed polynomially weighted sequence norm, whereas the coefficient of $w^{-k}$ is exactly $t$. These give the coefficient estimate.
\end{proof}

For $k\ge2$ the first-order terms of the positive NP eigenvalues, counted with multiplicity, are
\[
 \frac t2\sqrt{j(k+1-j)},\qquad j=1,\ldots,k,
\]
after rearrangement in nonincreasing order. This is the first-order spectral pattern of \cite[Theorems~2.2--2.3]{APRYZ}. The cluster argument handles the repetitions and gives a uniform inverse estimate for exact eigenvalues. The mode $k$ and the remainder are fixed. At $t=0$, all mode families meet at the disk, so an unknown integer $k$ cannot be recovered uniformly.

\subsection{Comparison with plasmonic particle shape recovery}\label{sec:APRYZ}
Ammari, Putinar, Ruiz, Yu, and Zhang \cite{APRYZ} study a single-particle family with normalised boundary $e^{i\theta}+t e^{-ik\theta}$ and one continuous coefficient for each fixed mode index. Our main theorem varies $m$ real coefficients independently, with a prescribed higher tail. Their finite Laurent maps generally lie outside the positive-orthant cones of the $m$-parameter theorem. Theorem~\ref{thm:singlemode} treats their repeated leading singular values separately.

Their Theorems~2.2--2.3 give first-order NP eigenvalue expansions and leading-order eigenfunctions. For $k\ge3$, Section~2.6 identifies two distinct positive leading resonance locations from the first-order polarization tensor and second-order contracted polarization tensors:
\[
 \ell=\tfrac t2\sqrt{k},\qquad
 \ell'=\tfrac t2\sqrt{2(k-1)}.
\]
These give the algebraic inversion
\[
 k=\frac{\ell^2}{\ell^2-(\ell')^2/2},\qquad
 t=2\sqrt{\ell^2-(\ell')^2/2}.
\]
For $k=1$, $\ell'=0$; for $k=2$, the two locations coincide. For $k\ge3$, the locations approximate exact NP eigenvalues to first order, so the formulas invert the leading model. The response channels label the resonances through geometric couplings; they need not be the first two eigenvalues ordered by size. Nor need a finite collection of leading visible resonances exhaust the exact spectrum.

Our injectivity and uniform Lipschitz estimates on specified small families rely on controlling the nonlinear remainder and counting eigenvalues with multiplicity. The response analysis in \cite{APRYZ} identifies the modes that couple to low-order polarization measurements and recovers size from amplitudes in the leading model, adding information beyond eigenvalues. Its Theorem~3.2 gives an exact spectrum for two nearly touching disks. For equal radius $r>0$ and gap $g>0$, the ordered positive eigenvalues on densities with zero charge on each component satisfy
\[
 \lambda_{+,2n-1}=\lambda_{+,2n}=\tfrac12e^{-2n\xi_0},\quad n\ge1,\qquad
 \cosh\xi_0=1+\frac{g}{2r}.
\]
Each positive and negative level has multiplicity two, corresponding to the cosine and sine bipolar modes. Hence $g/r=2\{\cosh[-\tfrac12\log(2\lambda_{+,1})]-1\}$. The relative gap is thus recovered exactly in this disconnected geometry; absolute size remains undetermined.

\subsection{A local criterion and prescribed perturbations}
\begin{theorem}[Finite-dimensional local spectral criterion]\label{thm:localslice}
Let $p\mapsto D_p$, $p\in U$, be a normalised family on an open set $U\subset\R^d$, and let $\mathcal F:U\to\R^N$ be a $C^1$ vector of quantities determined by the NP spectrum, with $N\ge d$. Suppose $p_0\in U$ and $D\mathcal F(p_0)$ has rank $d$. There is then a convex ball $V$ about $p_0$ on which
\[
 c|p-q|\le|\mathcal F(p)-\mathcal F(q)|\le C|p-q|.
\]
If the boundary parametrisations are locally Lipschitz in $p$ in $C^0$, then
$d_H(\partial D_p,\partial D_q)\le C'|\mathcal F(p)-\mathcal F(q)|$.
\end{theorem}
\begin{proof}
Let $A=D\mathcal F(p_0)$ and let $\sigma>0$ be its smallest singular value. On a sufficiently small ball, $\norm{D\mathcal F(p)-A}_{\op}\le\sigma/2$. Integrating along the segment from $q$ to $p$ gives
\[
 \mathcal F(p)-\mathcal F(q)=A(p-q)+e,
 \qquad |e|\le\tfrac\sigma2|p-q|.
\]
The smallest singular value of $A$ gives the lower bound; boundedness of $D\mathcal F$ gives the upper bound. Pairing the same boundary parameter on the two curves gives the Hausdorff estimate.
\end{proof}
One may use isolated simple eigenvalue branches or sums of powers over isolated clusters containing multiple eigenvalues. The rank condition must be checked in each case. With normalisation fixed, the criterion also covers nonsymmetric conformal families and families with complex coefficients, provided the spectral differential has full real rank.

The criterion persists under prescribed perturbations. If $\norm{D\mathcal F-A}_{\op}\le\sigma/4$ on a convex parameter set $V$, then every $C^1$ spectral map $\widetilde{\mathcal F}$ with $\norm{D\widetilde{\mathcal F}-D\mathcal F}_{\op}\le\sigma/4$ has lower Lipschitz constant $\sigma/2$. A prescribed perturbation satisfying these bounds therefore preserves uniqueness; the prescribed quadratic remainder in Corollary~\ref{cor:generaltail} gives such an extension near the disk. Allowing tail coefficients to vary independently adds unknowns. Differentiating the lower Lipschitz bound along each direction shows that a differentiable spectral map locally bi-Lipschitz in $d$ parameters must have differential of rank $d$. The argument therefore needs at least $d$ scalar observables. Uniqueness on a union of families also requires comparisons between the families.

\subsection{M\"obius saturation and quotient stability}\label{sec:mobius}
M\"obius transformations generally change the normalisation, the prescribed Laurent coefficient relations, and convexity. They still give an invariant enlargement of a normalised uniqueness class.

For a smooth Jordan domain $D$ and an extended M\"obius map $T$ with no pole on $\partial D$, let $T_\#D$ denote the bounded Jordan domain enclosed by $T(\partial D)$. Given a normalised uniqueness class $\mathcal A=\{D_{\avec}:\avec\in V\}$ with $V\subset\R^d$, define its saturation by
\[
 \operatorname{Sat}(\mathcal A)=\{T_\#D_{\avec}:\avec\in V,\ T\text{ extended M\"obius, pole off }\partial D_{\avec}\}.
\]
\begin{lemma}[Reduced planar spectral invariance]\label{lem:mobius}
For every smooth Jordan domain $D$ and every such $T$, the reduced NP spectra of $D$ and $T_\#D$ agree, including multiplicities.
\end{lemma}
\begin{proof}
For $T(z)=(az+b)/(cz+d)$,
\[
 \log|T(x)-T(y)|=\log|x-y|+\log|ad-bc|
 -\log|cx+d|-\log|cy+d|.
\]
Define the charge-preserving boundary transport by
$(U\varphi)(T(x))=\varphi(x)/|T'(x)|$. For mean-zero densities, the normal derivative of the extra $x$-term vanishes after integration, since its coefficient is total charge; the extra $y$-term is independent of $x$. The conformal normal-derivative rule gives
\[
 K^*_{T_\#D}U=\epsilon U K_D^*,
\]
where $\epsilon=1$ when $T(D)$ is bounded, and $\epsilon=-1$ when the bounded complementary component is selected. The transport $U$ is boundedly invertible between the corresponding Sobolev spaces and preserves zero total charge. The reduced spectra therefore agree up to sign. Their planar Jordan symmetry $\mu\leftrightarrow-\mu$ removes the sign. Reflection preserves the kernel, giving the extended case by composition.
\end{proof}

\begin{theorem}[Uniqueness and stability modulo M\"obius transformations]\label{thm:saturation}
Suppose $\mathcal A$ has injective spectral data $\Lambda(D_{\avec})$ and
\[
 |\avec-\bvec|\le C|\Lambda(D_{\avec})-\Lambda(D_{\bvec})|.
\]
Then the NP spectrum determines each domain in $\operatorname{Sat}(\mathcal A)$ up to extended M\"obius equivalence. Every orbit in this saturation has exactly one representative in $\mathcal A$. On the quotient define
\[
 d_{\mathcal A}([T_\#D_{\avec}],[S_\#D_{\bvec}])=|\avec-\bvec|.
\]
This metric is well defined and satisfies the same inverse Lipschitz bound in the spectral data. If the original class also has the forward spectral estimate, the quotient parameter metric and spectral distance are bi-Lipschitz equivalent.
\end{theorem}
\begin{proof}
Equal spectra for $T_\#D_{\avec}$ and $S_\#D_{\bvec}$ imply equal spectra for $D_{\avec},D_{\bvec}$ by Lemma~\ref{lem:mobius}. Injectivity gives $\avec=\bvec$, proving uniqueness modulo the group. Each orbit has at most one representative in $\mathcal A$ by the same argument, and has one by the definition of saturation. The metric inherits the estimates on these representatives.
\end{proof}
The metric measures the parameters of the unique normalised representative in the chosen class. Eigenvalues cannot recover Euclidean positions or transformation parameters: translations preserve the spectrum, and taking an infimum of Hausdorff distances over arbitrary transformations of both domains gives zero under small scalings. For a \emph{known common} transformation $T$, Euclidean estimates follow from bounds on $|T'|$ and its inverse near the boundaries. The constants deteriorate as a pole approaches a boundary.

\subsection{Overlap with algebraic domains}\label{sec:algebraic}
A domain has a real algebraic boundary if its boundary is contained in the zero set of a nonzero polynomial in $\R[x,y]$.

\begin{proposition}[Overlap with algebraic boundaries]
Let $\Psi(w)=w+\sum_{n=1}^m a_nw^{-n}$ have real coefficients and be univalent on $\{|w|\ge1\}$, with $\Psi'(w)\ne0$ for $|w|=1$. The bounded domain enclosed by $\Psi(\{|w|=1\})$ then has a real algebraic boundary. Thus, with $F=0$, our uniqueness classes consist of smooth algebraic Jordan domains. Algebraicity alone does not imply membership in a fixed cone with a prescribed tail.
\end{proposition}
\begin{proof}
For $|w|=1$, define
\[
 z=\Psi(w),\qquad \zeta=\Psi(1/w)=\bar z.
\]
The rational functions $z(w),\zeta(w)$ are algebraically dependent because $\C(w)$ has transcendence degree one. Hence a nonzero polynomial $A(z,\zeta)$ vanishes on the curve. The invertible substitution $z=x+iy,\zeta=x-iy$ gives a nonzero complex polynomial in $x,y$; at least one of its real and imaginary parts is nonzero and gives a real polynomial vanishing on the boundary. This proves algebraicity. Univalence and the nonvanishing derivative give an analytic Jordan boundary.

Our families with prescribed tails impose additional coefficient relations and a cone restriction. For example, the comparable-coefficient $m=2$ sector excludes ellipses, whose normalised maps have only $a_1$ nonzero. The boundaries in Section~\ref{sec:counter} give a stronger obstruction. Both are algebraic of degree fourteen, with equations $|P(x+iy)|^2-R^2=0$ and $|Q(x+iy)|^2-R^2=0$. Their equal spectra and failure of extended-M\"obius equivalence prevent both domains from belonging to the same normalised uniqueness class, even modulo these transformations.
\end{proof}

The results in \cite{APSTident,APSTrecon} based on finitely many GPTs impose polynomial normalisation, contrast, and nondegeneracy conditions in addition to algebraicity. A prescribed analytic tail need not give an algebraic boundary. Conversely, an algebraic boundary may have an exterior map with infinitely many nonzero Laurent coefficients.

\subsection*{Declaration of AI use}
The authors interacted with and guided ChatGPT~6 whilst developing the proofs and then wrote the paper. The authors take full responsibility for the accuracy and integrity of all content in this paper.

\phantomsection

\end{document}